\documentclass{article}
\usepackage[utf8]{inputenc}
\usepackage{todonotes}
\usepackage{caption}
\usepackage[mathcal]{euscript}
\usepackage{nicefrac}
\usepackage{amsmath,enumerate}
 \usepackage{amsthm}
\usepackage{amsfonts}
\usepackage{amssymb}
\usepackage{mathrsfs}
\usepackage{pgfplots}
\usepackage{mathtools}
\usepackage{eufrak}
\usepackage{comment}
\usepackage{lmodern}
\usepackage{bm}
\usepackage{subfig}
\usepackage{caption}
\usepackage{cite}
\usepackage[margin=1.5in]{geometry}
\usepackage{tikz}
\usetikzlibrary{shapes,arrows,graphs,graphs.standard,shapes.misc}
\usetikzlibrary{matrix,decorations.pathreplacing, calc, positioning,fit}
\usetikzlibrary{shapes.geometric}

\usepackage{tikz}
\usetikzlibrary{external}
\usetikzlibrary{decorations.shapes,arrows,calc,decorations.markings,shapes,decorations.pathreplacing,shapes.arrows}
\usetikzlibrary{positioning}
\tikzset{main node/.style={circle,fill=blue!20,draw,inner sep=1pt},}

\newtheorem{theorem}{Theorem}[section]
\newtheorem{definition}{Definition}

\newtheorem{lemma}{Lemma}

\newtheorem{corollary}[theorem]{Corollary}
\newtheorem{proposition}[theorem]{Proposition}

\renewcommand{\endproof}{\hfill $\blacksquare$}
\newcommand{\R}{\mathbb{R}}
\newcommand{\C}{\mathbb{C}}
\newcommand{\F}{\mathbb{F}}
\newcommand{\argmax}{\operatorname{argmax}}
\newcommand{\eig}{\operatorname{eig}}

\newcommand{\argmin}{\operatorname{argmin}}
\newcommand{\rd}{\mathrm{d}}

\title{Controllability Issues of Linear Ensemble Systems \\
over Multi-dimensional Parameterization Spaces}

\author{Xudong Chen\footnote{ECEE Department, CU Boulder. Email: \texttt{xudong.chen@colorado.edu}.}}

\begin{document}

\date{}
\maketitle 

\begin{abstract}
We address an open problem in ensemble control: Whether there exist controllable linear ensemble systems over multi-dimensional parameterization spaces?  
We provide a negative result: Any real-analytic linear ensemble system is not $\mathrm{L}^p$-controllable, for $2\le p \le \infty$, if its parameterization space contains an open set in $\R^d$ for $d \ge 2$.    
\end{abstract}

\section{Introduction and Main Result}
Ensemble control originated from quantum spin systems~\cite{glaser1998unitary,brockett2000stochastic,li2006control} and has found many applications across various disciplines in science and engineering, ranging from neuroscience~\cite{ching2013control,zlotnik2016phase,mardinly2018precise}, to emergent behaviors~\cite{brockett2010control}, and to multi-agent control~\cite{becker2012approximate,becker2017controlling,chen2019controllability}. 
Driven by these emerging applications, there has been an active development in mathematical control theory for analyzing basic properties of infinite ensemble systems, among which controllability has been a major focus.   
Although significant progress has been made over the last score, a complete understanding of controllability is still lacking. This is true even for ensembles of linear control systems.   
In the paper, we consider  ensembles of linear time-invariant systems whose $(A,B)$ pairs are continuous, matrix-valued functions defined on compact subsets of multi-dimensional Euclidean spaces. We call these subsets {\em parameterization spaces}. We address controllability issues of those linear ensemble systems.  

\subsection{Successes in one dimension} 
When parameterization spaces are one-dimensional, it is known that there exist uniformly controllable linear ensemble systems. 
 We take below a simple but illustrative example: Consider a scalar linear ensemble system over the closed unit interval $[0,1]$: 
\begin{equation}\label{eq:exampleone}
\dot x(t, \sigma)  := \frac{\partial}{\partial t} x(t,\sigma) = \sigma x(t, \sigma) + u(t), \quad \mbox{for all }\sigma\in [0,1], 
\end{equation} 
where $x(t, \sigma)\in \R$ is the current state of an individual system indexed by $\sigma$, and $u(t)\in \R$ is the  control input common to all individual systems. For a fixed time~$t$, the collective of $x(t,\sigma)$, for $\sigma\in [0,1]$, is called a {\em profile},  which we denote by $\chi(t)$. The profile $\chi(t)$ can be viewed as a function $\chi(t): \Sigma \to \R$, sending $\sigma$ to $x(t,\sigma)$. For this example, we assume that profiles are continuous functions. Then, uniform controllability of system~\eqref{eq:exampleone} is, roughly speaking, the capability of using the scalar control input $u(t)$ to steer from an arbitrary initial profile $\chi(0)$ to reach a profile which is $\mathrm{L}^\infty$-close to a given, but arbitrary, target profile $\hat \chi$.

In this setup, the controllable subspace associated with system~\eqref{eq:exampleone} is the uniform closure of the vector space (over $\R$) spanned by $a^kb$, for $k\ge 0$, where the associated $(a, b)$ pair is given by $a(\sigma) := \sigma$ and $b(\sigma) := 1$ for all $\sigma\in [0,1]$. 
System~\eqref{eq:exampleone} is uniformly controllable if and only if the controllable subspace comprises continuous functions from $[0,1]$ to $\R$ (see Lemma~\ref{lem:controllablesubspace} below).   
Here, $a^kb$ are simply the monomials~$\sigma^k$. By Stone-Weierstrass theorem~\cite[Ch.~7]{rudin1976principles}, any continuous function on $[0,1]$ can be approximated uniformly and arbitrarily well by polynomials. Thus, system~\eqref{eq:exampleone} is uniformly controllable.

  Significant extensions of the above controllability result have been made over the last decade. Necessary and/or sufficient conditions have been established for controllability of general linear ensemble systems over single closed intervals~\cite{li2010ensemble,qi2013ensemble,helmke2014uniform,li2015ensemble,zeng2016moment}, finite unions of closed intervals~\cite{schonlein2016controllability,li2020separating}, and curves in the complex plane~\cite{dirr2021uniform}.  Although the analysis for a general case is much more involved, Stone-Weierstrass theorem (or Mergelyan's theorem~\cite[Ch.~20]{rudin1987real} for dealing with {\em complex} linear ensemble systems) is the core as was illustrated above.  
  We also refer the reader to~\cite{li2010ensemble,li2011ensemble,dirruniform,zhang2018controllability} for ensembles of linear time-variant systems and, further, to a book chapter~\cite[Ch.~12]{fuhrmann2015mathematics} for more relevant works. 
  
  Note that any compact, connected, one-dimensional manifold is homeomorphic to either a closed interval or a circle~\cite[Ch.~2]{guillemin2010differential}. The literature is relatively sparse for linear ensemble systems over circles: It is known~\cite[Remark 9-(d)]{dirr2021uniform} that there exist scalar, complex linear ensemble systems with single control inputs that are $\mathrm{L}^2$-controllable (see Definition~\ref{def:controllability} below); the conditions about cyclic operators and cyclic vectors described in~\cite{ross2009common} can be used to establish the fact. A negative result about uniform controllability for those linear ensemble systems has been established in~\cite[Lemma 5]{dirr2021uniform}.    
  
  \subsection{Problem for multi-dimensions} 
  Those existing results make us wonder whether the successes can be repeated if the dimensions of the parameterization spaces are increased?    
  This is in fact an open problem. 
  
  Before we provide a solution to the problem, perhaps it is helpful to gain some insights by looking into a complex version of system~\eqref{eq:exampleone}. Consider a linear ensemble system with the same dynamics as~\eqref{eq:exampleone}, but with $\sigma$ being a complex variable that belongs to the closed unit disk centered at the origin of the complex plane:       
\begin{equation}\label{eq:exampletwo}
\dot x(t, \sigma) = \sigma x(t, \sigma) + u(t), \quad \mbox{for all } \sigma\in \C \mbox{ and } |\sigma| \le 1.
\end{equation} 
The state $x(t, \sigma)$ is now complex-valued. 
We allow the scalar control input $u(t)$ to take complex value as well. Note that we choose to work with complex systems is for ease of presentation: One can obtain a corresponding real ensemble system  by realification. The state space of each individual system after realification will be two-dimensional. We elaborate on the correspondence later in Lemma~\ref{lem:equivalentRandC}, Section~\S\ref{sec:preliminaryresult}.

The controllable subspace associated  with~\eqref{eq:exampletwo} is, similarly, given by the uniform closure of the space (but now, over $\C$) spanned by all the monomials $\sigma^k$ for $k\geq 0$. However, unlike the previous case, what we obtain after taking the closure is not the space of continuous functions anymore. It follows from Mergelyan's theorem that the controllable subspace comprises functions that are holomorphic in the interior of the closed disk and continuous on the boundary. As a consequence, we lose uniform controllability of system~\eqref{eq:exampletwo}. 
One may wonder at the point whether we could fix the controllability issue by increasing the dimension of state space and/or by adding more control inputs? The answer is no; in this paper, we show that if the parameterization space $\Sigma$ contains an open set $U$ in $\R^d$ for~$d \geq 2$, and if $A$ and $B$ are real-analytic at a certain point in $U$, then the linear ensemble system described by the $(A, B)$ pair cannot be uniformly or $\mathrm{L}^p$ controllable, for $p\geq 2$.   

To the best of the author's knowledge, the negative result is original. Previous works on the problem have mainly focussed on obtaining necessary conditions for controllability. For example, Helmke and Sch\"onlein have provided in~\cite{helmke2014uniform} conditions about disjointness of the spectrums of the $A$-matrix.  
Later in~\cite{schonlein2016controllability}, the authors have also shown that if uniform controllability is concerned, then under some other mild assumptions, the parameterization space is at most two-dimensional and, moreover, the $A$-matrix cannot have a branch of real eigenvalues. In a more recent work~\cite{dirr2021uniform}, Dirr and Sch\"onlein have shown that if there is only one single control input, then a linear ensemble system cannot be uniformly controllable if the dimension of the parameterization space is greater than one. Example and simulation studies for linear ensemble systems over two-dimensional parameterization spaces are also carried out by Zhang and Li in~\cite{zhang2018controllability}.  

\subsection{Main results}\label{ssec:mainresult} 
Let $\Sigma$ be a compact subset of $\R^d$. 
Let $\F$ be the field of either real or complex numbers. 
We consider a continuum ensemble of linear time-invariant control systems over $\Sigma$:
	\begin{equation}\label{eq:model}
	\dot x(t, \sigma) = A(\sigma) x(t, \sigma) + B(\sigma)u(t), \quad \mbox{for all } \sigma\in \Sigma,
	\end{equation}    
	where $x(t,\sigma)\in \F^n$, $u(t)\in \F^m$, and $A:\Sigma \to \F^{n\times n}$ and $B:\Sigma\to\F^{n\times m}$ are continuous matrix-valued functions. 
	The control input $u(t)$ is said to be {\em admissible} if for any $T>0$, $u:[0,T]\to \F^m$ is integrable.

Let $\chi(t): \Sigma \to \F^n$ be the profile at time $t$, defined as the function sending $\sigma$ to $x(t, \sigma)$. In this paper, the profiles $\chi(t)$ are either continuous or $\mathrm{L}^p$-functions, for $1\leq p < \infty$. 
Denote by $\mathrm{C}^0(\Sigma, \F^n)$ the space of continuous functions $f: \Sigma\to \F^n$, and by  
	${\rm L}^p(\Sigma, \F^n)$, for $1\leq p < \infty$, the Banach space of all functions $f: \Sigma\to \F^n$ whose ${\rm L}^p$-norm is finite. 
The {\em profile space} of system~\eqref{eq:model}, denoted by $\mathrm{X}^p_\F$, is given by
\begin{equation*}\label{eq:defXp}
\mathrm{X}^p_\F:= 
\left\{
\begin{array}{ll}
	\mathrm{L}^p(\Sigma, \F^n) & \mbox{if } 1\leq p < \infty, \\
	\mathrm{C}^0(\Sigma, \F^n) & \mbox{if } p = \infty.
\end{array}
\right.
\end{equation*}   
We now have the following definition:

\begin{definition}\label{def:controllability}
 	System~\eqref{eq:model}, or simply the pair $(A, B)$, is {\bf ${\rm L}^p$-controllable}\footnote{For $p=\infty$, ${\rm L}^\infty$-controllability is also known as {\em uniform controllability}~\cite{fuhrmann2015mathematics}.}, for $1\le p \le \infty$, if for any initial profile $\chi(0)\in \mathrm{X}^p_\F$, any target profile $\hat \chi\in \mathrm{X}^p_\F$, and any error tolerance $\epsilon > 0$, there is a time $T > 0$ and an admissible control input $u:[0,T]\to \F^m$ such that the solution $\chi(t)$ generated by~\eqref{eq:model} satisfies $\|\chi(T) - \hat \chi\|_{\rm L^p} < \epsilon$.    
 \end{definition}

Let $\sigma_0$ be a point of $\Sigma$.   	
A function~$f:\Sigma \to \R$ is said to be {\em real-analytic at~$\sigma_0$} if there exists an open neighborhood $U$ of $\sigma_0$ in $\R^d$ such that $f |_U$ can be represented by a convergent power series in the entries of~$(\sigma-\sigma_0)$. Note that if~$\sigma_0$ belongs to the boundary of $\Sigma$, then real-analyticity of~$f$ at $\sigma_0$ means that $f$ can be extended to an open neighborhood~$U$ of~$\sigma_0$ and the extended function is real-analytic at~$\sigma_0$. A complex-valued function $f:\Sigma\to \C$ is said to be {\em real-analytic at~$\sigma_0$} if both real and imaginary parts of~$f$ are real-analytic at~$\sigma_0$. A matrix-valued function is {\em real-analytic at~$\sigma_0$} if all of its entries are real-analytic at~$\sigma_0$.

We now state the main result of the paper:

\begin{theorem}\label{thm:main}  	 
   	If $\Sigma$ contains an open set $U$ in $\R^d$, with $d \geq 2$, and if continuous matrix-valued functions $A:\Sigma\to\F^{n\times n}$ and $B:\Sigma\to\F^{n\times m}$, with $\F = \R$ or $\F = \C$, are real-analytic at a certain point in~$U$, then the linear ensemble system $(A, B)$ cannot be $\mathrm{L}^p$-controllable, for $2\leq p \leq \infty$.  
\end{theorem}

Theorem~\ref{thm:main} can be formulated as a negative result in approximation theory. For that, we first have the following definition:

\begin{definition}
	Let $A:\Sigma\to \F^{n\times n}$ and $B:\Sigma \to \F^{n\times m}$ be continuous matrix-valued functions. 
	The {\bf $\mathrm{L}^p$-controllable subspace} of system~\eqref{eq:model}, denoted by $\mathcal{L}^p_\F(A, B)$, is the ${\rm L}^p$-closure of the subspace, over $\F$, spanned by the columns of $A^kB$, for all $k \ge 0$.   
\end{definition}

The above definition is a straightforward generalization of the controllable subspace associated with a finite dimensional linear system. By the Kalman rank condition, a finite-dimensional linear system is controllable if and only if the controllable subspace is the entire state space.  This is, in fact, true for linear ensemble systems.  
We introduce below a necessary and sufficient condition for ${\rm L}^p$-controllability  adapted from~\cite{triggiani1975controllability}:

\begin{lemma}\label{lem:controllablesubspace}
	System~\eqref{eq:model} is $\mathrm{L}^p$-controllable if and only if 
	$
	\mathcal{L}^p_\F(A, B)= \mathrm{X}^p_\F
	$.  
\end{lemma}

With Lemma~\ref{lem:controllablesubspace}, the following result is then equivalent to Theorem~\ref{thm:main}:

\begin{theorem}\label{thm:approximation}
If $\Sigma$ contains an open set $U$ in $\R^d$, with $d \geq 2$, and if continuous matrix-valued functions $A:\Sigma\to\F^{n\times n}$ and $B:\Sigma\to\F^{n\times m}$ are real-analytic at a certain point in~$U$, then $\mathcal{L}^p_\F(A, B)$ is a proper subspace of $\mathrm{X}^p_\F$, for $2\leq p\leq \infty$. 
\end{theorem}

Now, let $\Sigma$ be a compact subset of $\C$. Given a bounded function $a:\Sigma\to \C$, let $\mathcal{M}_a: \mathrm{L}^p(\Sigma, \C)\to \mathrm{L}^p(\Sigma, \C)$, for $1\leq p <\infty$, be the {\em multiplication operator}, defined as 
	$\mathcal{M}_a(b):= ab$. 
	The operator $\mathcal{M}_a$ is said to be {\em cyclic}~\cite{seid1974cyclic} if there exists an $\mathrm{L}^p$-function $b$ such that $\mathcal{L}_\C^p(a,b) = \mathrm{L}^p(\Sigma,\C)$, and any such $b$ is called a {\em cyclic function} with respect to $\mathcal{M}_a$. 
	Denote by $\omega: \Sigma \to \C$ the identity function, i.e., $\omega(\sigma) = \sigma$ for all $\sigma\in \Sigma$. 
	It is known~\cite{bram1955subnormal,seid1974cyclic} that $\mathcal{M}_\omega$ is a cyclic operator and, moreover, it is a canonical one in a sense that if $\mathcal{M}_{a}$ is another cyclic operator, then there exists a surjective isometry $J: \mathrm{L}^p(\Sigma,\C)\to \mathrm{L}^p(\Sigma,\C)$ such that $\mathcal{M}_{a} = J^{-1} \mathcal{M}_{\omega} J$.   
	While cyclic operators have been characterized and understood to a certain extent, it still remains open what type of elements $b\in \mathrm{L}^p(\Sigma,\C)$ can be cyclic functions.  
	A necessary and sufficient condition has recently been obtained in~\cite[Proposition 7]{dirr2021uniform}, yet there is still lack of an explicit and complete characterization.   
	Our contribution to this area is formulated in the following result, which is an immediate consequence of  Theorem~\ref{thm:approximation}: 

\begin{corollary}\label{cor:cyclicoperator}
	Let $\Sigma$ be a compact subset of $\C$ and $\omega:\Sigma\to \C$ be the identity function. 
	Suppose that $b\in \mathrm{L}^p(\Sigma, \C)$, for $2\leq p < \infty$, is a cyclic function with respect to $\mathcal{M}_\omega$; then, $b$ is nowhere real-analytic over the interior of $\Sigma$.  
\end{corollary}

\subsection{Organization of the paper}    	
The remainder of the paper is devoted to the proof of Theorem~\ref{thm:main}. The proof is divided into three parts:      	

In Section~\S\ref{sec:preliminaryresult}, we present preliminary results that can reduce moderately the complexity of controllability analysis for system~\eqref{eq:model}. 
By the end of Section~\S\ref{sec:preliminaryresult}, we will be able to focus only on $\mathrm{L}^2$-controllability of complex linear ensemble systems over closed, $d$-dimensional balls. 
     
In Section~\S\ref{sec:normalform}, we introduce a special class of (complex) linear ensemble systems, termed {\em normal forms}. Each normal form is a scalar complex linear ensemble system whose parameterization space is a closed two-dimensional disk. Moreover, the $A$-matrix, now being a scalar, is the identity function and the $B$-matrix, now being a row vector, is real-analytic. We show that every normal form is not $\mathrm{L}^2$-controllable. 
 
In Section~\S\ref{sec:proofofmain},	we bridge the gap between $\mathrm{L}^2$-controllability of normal forms and $\mathrm{L}^2$-controllability of general linear ensemble systems~\eqref{eq:model}. The analysis will be carried out by a sequence of reductions on both state spaces of individual systems  and parameterization spaces. 
After these reductions, we will be able to focus only on scalar complex linear ensemble systems over two-dimensional disks. These systems will be further translated into the normal forms with controllability preserved. 
All the arguments then form a complete proof of Theorem~\ref{thm:main}.

\subsection{Notations}  For a complex number $z = x + \mathrm{i} y$, let $\bar z = x - \mathrm{i} y$ be the complex conjugate of $z$. 
The polar representation of $z$ is given by $z = re^{\mathrm{i}\theta}$ where $r\ge 0$ and $\theta\in [-\pi,\pi)$. If $Z$ is a complex matrix, then $\overline{Z}$ is entry-wise, and we let $Z^\dagger := \overline{Z}^\top$.

Let $S$ be a subset of $\C^n$. 
A function $f: S\to \C^n$ is said to be $\mathrm{C}^k$, for $k\geq 0$, real-analytic, or holomorphic if it can be extended to a $\mathrm{C}^k$, real-analytic, or holomorphic function on an open set $S'$ that contains~$S$ (if $S$ is open, then $S'$ can simply be $S$).

Throughout the paper, we use $\omega_S: S\to S$ to denote the identity function, i.e., $\omega_S(x) = x$ for all $x\in S$.  
We let ${\bf 1}_S: S\to \R$ be the constant function that takes value one everywhere, i.e., ${\bf 1}_S(x) = 1$ for all $x\in S$. 
For ease of notation, we will omit sometimes the subindex $S$ and simply write $\omega$ and ${\bf 1}$.

Let $S$ be a Lebesgue measurable subset of $\R^n$. Let $f_1$ and $f_2$ be two complex, vector-valued, square-integrable functions defined on $S$. We define the inner-product of $f_1$ with $f_2$ as 
$
\langle f_1, f_2\rangle_{S}:= \int_{S} f^\dagger_1(\sigma)f_2(\sigma) \mathrm{d}\sigma
$.  
Note that $\langle f_1, f_2\rangle_{S} = \overline{\langle f_2, f_1\rangle}_{S}$. We will omit the subindex $S$ if it does not cause any confusion.

\section{Preliminary Results}\label{sec:preliminaryresult}
In this section, we will {\em (1)} establish equivalence of controllability for real and complex linear ensemble systems; {\em (2)} compare $\mathrm{L}^p$-controllability for different values of~$p$; and {\em (3)} introduce ensemble systems obtained by pullbacks and relate controllability properties of these systems to those of the original ones~\eqref{eq:model}. 
The results are formulated as Lemmas~\ref{lem:equivalentRandC}--\ref{lem:embedding} and presented in the subsequent subsections.

\subsection{Controllability of real and complex ensembles}  
As indicated at the beginning of Subsection~\S\ref{ssec:mainresult}, the field $\F$ can be either $\R$ or $\C$. When $\F = \C$ (resp. $\F = \R$), we call system~\eqref{eq:model} a {\em complex} (resp. {\em real}) linear ensemble system. Since $\R \subset \C$, the pair $(A, B)$ associated with a complex linear ensemble system can be real, matrix-valued functions (but the control input $u(t)$ can be valued in $\C^m$).  

We have the following result: 

\begin{lemma}\label{lem:equivalentRandC}
	There is a complex ${\rm L}^p$-controllable linear ensemble system if and only if there is a real ${\rm L}^p$-controllable one. 
\end{lemma}

\begin{proof}
	If system~\eqref{eq:model} is real and ${\rm L}^p$-controllable, then it is known (see, e.g.,~\cite[Lemma~1]{dirr2021uniform} and~\cite{fattorini1966some}) that the same pair $(A, B)$ yields a complex, ${\rm L}^p$-controllable linear ensemble system.
		We now assume that system~\eqref{eq:model} is complex and ${\rm L}^p$-controllable. We show below that its realification is ${\rm L}^p$-controllable. 
	First, decompose $A = A_1 + \mathrm{i} A_2$ and $B = B_1 + \mathrm{i} B_2$ into real and imaginary parts. The realification of~\eqref{eq:model} is then a $2n$-dimensional real linear ensemble system given as follows: 
\begin{equation}\label{eq:correspondingrealsys}
\begin{bmatrix}
\dot x_{1}(t,\sigma) \\
\dot x_{2}(t,\sigma)
\end{bmatrix}
= 
\begin{bmatrix}
	A_1(\sigma) & -A_2(\sigma)\\
	A_2(\sigma) & A_1(\sigma)
\end{bmatrix}
\begin{bmatrix}
x_{1}(t,\sigma) \\
x_{2}(t,\sigma)
\end{bmatrix}
 + 
\begin{bmatrix}
	B_1(\sigma) & -B_2(\sigma) \\
	B_2(\sigma) & B_1(\sigma)
\end{bmatrix} 
\begin{bmatrix}
	u_1(t) \\
	u_2(t)
\end{bmatrix}.
\end{equation}
The correspondence between~\eqref{eq:model} and~\eqref{eq:correspondingrealsys} is  straightforward: 
The two $n$-dimensional substates $x_{1}(t,\sigma)$ and $x_{2}(t,\sigma)$ in~\eqref{eq:correspondingrealsys} correspond to the real and imaginary parts, respectively, of $x(t,\sigma)$ in~\eqref{eq:model}. The same holds for $u_1(t)$ and $u_2(t)$, i.e., they are real- and imaginary-parts of $u(t)$ in~\eqref{eq:model}.  
We conclude from Definition~\ref{def:controllability} that if the linear complex ensemble system~\eqref{eq:model} is $\mathrm{L}^p$-controllable, then so is its realification~\eqref{eq:correspondingrealsys}.  
 \end{proof}

In the sequel, we will let $\F = \C$, i.e., we will consider {\em complex} linear ensemble systems. 
The choice is made for ease of analysis.  
For ease of notation, we will simply write $\mathcal{L}^p(A, B)$ by omitting its subindex $\C$. 

\subsection{Comparison between different notions of controllability}  
We have the following result that compares ${\rm L}^p$-controllability for different values of~$p$: 

\begin{lemma}\label{lem:l2controllability}
	If system~\eqref{eq:model} is $\mathrm{L}^p$-controllable and if $1\leq q < p \leq \infty$, then the system is also $\mathrm{L}^q$-controllable.  
\end{lemma}

\begin{proof}
	First, note that $\mathrm{L}^p(\Sigma, \C^n)$ is a subset of $\mathrm{L}^q(\Sigma, \C^n)$; indeed, by the H\"older's inequality, we have that $\|f\|_{\mathrm{L}^q} \leq \|f\|_{\mathrm{L}^p} \operatorname{vol}(\Sigma)^{\frac{1}{q} -\frac{1}{p}}$ for any $f\in \mathrm{L}^p(\Sigma, \C^n)$, where $\operatorname{vol}(\Sigma)$ is the volume of $\Sigma$. 
	It follows that $\|f\|_{\mathrm{L}^q}$ is finite and, hence, $f\in \mathrm{L}^q(\Sigma, \C^n)$. By the same argument, we know that $\mathcal{L}^q(A, B)$ contains $\mathcal{L}^p(A, B)$ as a subset.  
Because system~\eqref{eq:model} is ${\rm L}^p$-controllable, by Lemma~\ref{lem:controllablesubspace}, $\mathcal{L}^p(A, B)$ (and, hence, $\mathcal{L}^q(A, B)$) contains $\mathrm{C}^0(\Sigma, \C^n)$ as a subset. Since $\Sigma$ is compact, $\mathrm{C}^0(\Sigma, \C^n)$ is dense in $\mathrm{L}^q(\Sigma, \C^n)$ with respect to the $\mathrm{L}^q$-norm. Finally, note that $\mathcal{L}^q(A, B)$ is closed, so $\mathcal{L}^q(A, B) = \mathrm{L}^q(\Sigma, \C^n)$. By Lemma~\ref{lem:controllablesubspace}, system~\eqref{eq:model} is $\mathrm{L}^q$-controllable. 
\end{proof}

By Lemma~\ref{lem:l2controllability}, if system~\eqref{eq:model} is not $\mathrm{L}^2$-controllable, then it cannot be $\mathrm{L}^p$-controllable for all $p\ge 2$. Thus, to prove Theorem~\ref{thm:main}, it suffices to prove for the case where $p = 2$. 
Because of this, we assume, in the sequel, that $p = 2$.  
For ease of notation, we will write 
$\mathcal{L}(A, B):= \mathcal{L}^2(A, B)$
by omitting the sup-index. We will also omit, on occasions, the prefix ``$\mathrm{L}^2$-'' for controllability. For example, we will write controllable subspace instead of $\mathrm{L}^2$-controllable subspace.  

\subsection{Pullbacks by embeddings and subensembles} 
In this subsection, we assume that $\Sigma$ contains an open set $U$ in $\R^d$. 
Let $\Sigma'$ be a closed, $d$-dimensional ball (or a rectangle) in~$\R^d$, and $\varphi: \Sigma' \to U$ be a $\mathrm{C}^1$-embedding. 
Let $A':\Sigma'\to \C^{n\times n}$ and $B':\Sigma' \to \C^{n\times m}$ be defined as $A':= A\cdot \varphi$ and $B':= B\cdot \varphi$. 
We consider the following ensemble system:  
\begin{equation}\label{eq:embedding}
\dot x'(t, \sigma') = A'(\sigma') x'(t, \sigma') + B'(\sigma') u'(t), \quad \mbox{for all } \sigma'\in \Sigma',
\end{equation}
and have the following definition:

\begin{definition}\label{def:pullback}
System~\eqref{eq:embedding} is the {\bf pullback} of system~\eqref{eq:model} by~$\varphi$.  In the case	$\varphi:\Sigma'\to U$ is an inclusion map  we call system~\eqref{eq:embedding} a {\bf subensemble} or, more explicitly, {\bf subensemble-${\bf \Sigma'}$} of system~\eqref{eq:model}. 
\end{definition}

The following result relates controllability of system~\eqref{eq:model} to controllability of its pullback~\eqref{eq:embedding} (a similar result is obtained in~\cite[Lemma~1]{dirr2021uniform} for $\varphi$ an inclusion map):   
   	
\begin{lemma}\label{lem:embedding}
   	If system~\eqref{eq:embedding} is not controllable, then neither is system~\eqref{eq:model}. 	
\end{lemma}

\begin{proof}
Assuming that system~\eqref{eq:embedding} is not controllable, 
we will show that there exist a function $f\in \mathrm{L}^2(\Sigma, \C^n)$ and an $\epsilon > 0$ such that $f$ is at least $\epsilon$-away from $\mathcal{L}(A, B)$.

For any given $\sigma'\in \Sigma'$, we let $\rd\varphi_{\sigma'}: \R^d \to \R^d$ be the derivative of $\varphi$ at $\sigma'$. 
Because~$\varphi$ is an embedding, $\rd\varphi_{\sigma'}$ is a linear isomorphism. Thus, $\det(\rd\varphi_{\sigma'})$ is nonzero. Since $\varphi$ is $\mathrm{C}^1$ and since $\Sigma'$ is compact, there exist positive numbers $\kappa_0$ and $\kappa_1$ such that  
$\kappa_0 \leq |\det(\rd\varphi_{\sigma'})| \leq \kappa_1$, for all $\sigma'\in \Sigma'$. 

Since system~\eqref{eq:embedding} is not controllable, by Lemma~\ref{lem:controllablesubspace},  $\mathcal{L}(A', B')$ is a proper subspace of $\mathrm{L}^2(\Sigma', \C^n)$. Thus, there exist  a function $f'\in \mathrm{L}^2(\Sigma', \C^n)$ and an $\epsilon' > 0$ such that $f'$ is at least $\epsilon'$-away from $\mathcal{L}(A', B')$. 
Now, let $f: \Sigma \to \C^n$ be defined as follows:
$$
f(\sigma) := 
\left\{
\begin{array}{ll}
	f'(\sigma') & \mbox{if $\sigma = \varphi(\sigma')$ for some $\sigma'\in \Sigma'$},\\
	0 & \mbox{otherwise}.
\end{array}
\right.
$$ 
It follows from computation that $\|f\|_{\mathrm{L}^2} \leq \kappa_1 \|f'\|_{\mathrm{L}^2}$, so $f\in \mathrm{L}^2(\Sigma,\C^n)$.  

Given an arbitrary $g$ in $\mathcal{L}(A, B)$, let $g':\Sigma' \to \C^n$ be defined as $g'(\sigma') := g(\varphi(\sigma'))$. It should be clear that $g'\in \mathcal{L}(A',B')$. 
Moreover, we have that
\begin{align*}
\|g - f\|^2_{\mathrm{L}^2} \geq \| (g - f) |_{\varphi(\Sigma')}\|^2_{\mathrm{L}^2} \geq \kappa_0 \|g' - f'\|^2_{\mathrm{L}^2}  \geq \kappa_0 {\epsilon'}^2. 
\end{align*}
Thus, $f$ is at least $\sqrt{\kappa_0}\epsilon'$-away from $\mathcal{L}(A, B)$, which implies that $\mathcal{L}(A, B)$ is a proper subspace of $\mathrm{L}^2(\Sigma, \C^n)$. Thus, by Lemma~\ref{lem:controllablesubspace}, $(A, B)$ is not controllable. 
\end{proof}

If $A$ and $B$ are real-analytic at a certain point $\sigma_0\in U$, then they are real-analytic over an open neighborhood of $\sigma_0$, and any such open neighborhood contains a closed $d$-dimensional ball.  

Thanks to Lemma~\ref{lem:embedding}, we can now focus on the case where $\Sigma$ is itself a closed $d$-dimensional ball  and, moreover, $A:\Sigma\to\C^{n\times n}$ and $B:\Sigma\to\C^{n\times m}$ are real-analytic functions. 
However, even for such simplified case, the proof of Theorem~\ref{thm:main} is nontrivial at all.

\section{Normal Forms}\label{sec:normalform}
 In this section, we focus on a special class of complex linear ensemble systems, which we term {\em normal forms}. 
Each normal form is a scalar ensemble system, and its parameterization space is a closed, two dimensional disk in $\R^2$.  
In the sequel, we identify $\R^2$ with the complex plane~$\C$, so a point $\sigma = (\sigma_1, \sigma_2) \in \R^2$ corresponds to a complex number $\sigma = \sigma_1 + \mathrm{i}\sigma_2$.  
Define a disk of radius $R$ as follows: 
 $$
 D_0[R]:= \left \{ \sigma \in \C \mid |\sigma | \le R \right \}. 
 $$  
The square bracket in $D_0[R]$ indicates that it is a closed disk and the subindex $0$ indicates that the disk is centered at~$0$. We now have the following definition:

\begin{definition}
A {\bf normal form} is a scalar, complex linear ensemble system:
\begin{equation}\label{eq:normalform}
\dot x(t, \sigma) = \sigma x(t, \sigma) + b(\sigma) u(t), \quad \mbox{for all } \sigma \in D_0[R],
\end{equation}	
where $b:\Sigma\to\C^{1\times m}$ is an arbitrary real-analytic, vector-valued function. 
\end{definition}

The goal of the section is to establish the following result: 
  
\begin{theorem}\label{thm:specialcomplex}
	 Every normal form~\eqref{eq:normalform} is not $\mathrm{L}^2$-controllable.    
\end{theorem}

{\em Outline of proof:}  
By Lemma~\ref{lem:controllablesubspace}, Theorem~\ref{thm:specialcomplex} will be established if we can show that $\mathcal{L}(\omega,b)$ is a {\em proper} subspace of ${\rm L}^2(D_0[R], \C)$, where $\omega$ denotes the identity function on $D_0[R]$.  
In particular, if there exists a nonzero $f_0\in \mathrm{L}^2(D_0[R],\C)$  perpendicular to every subspace $\mathcal{L}(\omega, b_i)$, for $i = 1,\ldots, m$, then $f_0$ is perpendicular to $\mathcal{L}(\omega, b)$, which implies that $\mathcal{L}(\omega, b) \subsetneq {\rm L}^2(D_0[R], \C)$.  

The above arguments indicate that one can translate the $\mathrm{L}^2$-controllability problem for normal forms into the following intersection problem: Given finitely, but arbitrarily, many real-analytic functions $b_i: D_0[R]\to \C$, for $i = 1,\ldots, m$, is the intersection $\cap_{i = 1}^m \mathcal{L}^\perp(\omega, b_i)$  always nontrivial, where $\mathcal{L}^\perp(\omega, b_i)$ is the subspace of ${\rm L}^2(D_0[R], \C)$ perpendicular to $\mathcal{L}(\omega, b_i)$? 
We show that the answer is affirmative; we borrow a terminology from topology and call such a property the {\em finite intersection property}. This property will be formulated as a theorem, Theorem~\ref{thm:intersection}, in Subsection~\S\ref{ssec:finiteintersection}.   

The proof of existence of a desired $f_0$ is constructive, and it will take several steps. 
First, we use polar coordinates (i.e., $\sigma = r e^{\mathrm{i}\theta}$) to express each $b_i$ as a doubly infinite series $b_i(r,\theta) = \sum_{k = -\infty}^\infty \rho_{i,k}(r) e^{\mathrm{i}k\theta}$. Similarly, we write $f_0(r,\theta) = \sum_{k = -\infty}^\infty \rho_{0,k}(r)e^{\mathrm{i}k\theta}$. We call $\rho_{0,k}$ the {\em radius components} of $f_0$ and require that they satisfy certain conditions introduced in Definition~\ref{def:uec} so that the series $f_0$ is uniformly and exponentially convergent.   
Since $f_0$ is uniquely determined by its radius components (and vice versa), to construct $f_0$, it suffices to construct $\rho_{0,k}$. 
We do so by first establishing a necessary and sufficient condition on $\rho_{0,k}$, termed the {\em null condition}, for the resulting series $f_0$ to be perpendicular to every $\mathcal{L}(\omega,b_i)$ for $i = 1,\ldots, m$. This is done in Subsection~\S\ref{ssec:core}. 
Then, in Subsection~\ref{ssec:proof}, we exhibit appropriate $\rho_{0,k}$ they satisfy the null condition and render $f_0(r,\theta) = \sum_{k = -\infty}^\infty \rho_{0,k}(r)e^{\mathrm{i}k\theta}$ a desired convergent series.

It is worth pointing out that the analysis outlined above will be carried out on a closed annulus $A_1$ inside $D_0[R]$, rather than the disk $D_0[R]$ itself.  Specifically, we restrict each $b_i$ to $A_1$, and construct a nonzero $f_0$ on $A_1$ perpendicular to every subspace $\mathcal{L}(\omega_{A_1}, b_i |_{A_1})$. 
One then extends $f_0$ to a nonzero function $\tilde f_0\in \mathrm{L}^2(D_0[R],\C)$ simply by letting $\tilde f_0$ be identically~$0$ on $D_0[R]\backslash A_1$; it should be clear that $\tilde f_0$ is perpendicular to the subspaces $\mathcal{L}(\omega, b_i)$. 
The reason of performing the above-mentioned restriction on the domain (from $D_0[R]$ to $A_1$) is that by our construction, the radius components $\rho_{0,k}$ of $f_0$ will take the form $\rho_{0,k}(r) = q_k(r) r^{-k}$, where $q_k$ are polynomials with degrees less than or equal to~$m$ (the construction will be given in   Proposition~\ref{prop:constructionoff0}). 
Thus, the functions $\rho_{0,k}(r)$, for $k> m$, may diverge as $r$ approaches~$0$ and, hence,  the series $f_0(r,\theta) = \sum_{k = -\infty}^\infty \rho_{0,k}(r)e^{\mathrm{i}k\theta}$ may not be convergent for $r$ sufficiently small.

\subsection{Regularization condition}\label{ssec:regularization}
In this subsection, we introduce  a condition  that regularizes the $b$-vector in the normal form~\eqref{eq:normalform}. 
We show that this condition can be assumed for free when proving Theorem~\ref{thm:specialcomplex} and will be of great use in the analysis. 
To state the condition, we first recall that  a real-analytic function $f:D_0[R]\to \C$ can be locally represented by a convergent power series ({\em Maclaurin series}) in $\sigma$ and $\bar \sigma$:
\begin{equation}\label{eq:powerseries}
f(\sigma) = \sum^\infty_{k = 0}\sum^\infty_{\ell = 0} c(k,\ell) \sigma^k \bar \sigma^\ell, \quad \mbox{for all $\sigma$ such that } |\sigma| < \delta,  
\end{equation}  
where the coefficients $c(k,\ell)$ are complex numbers with $k$ and $\ell$ indicating the powers of $\sigma$ and $\bar \sigma$, respectively. The {\em radius of convergence} is defined to be the supremum of~$\delta$ such that~\eqref{eq:powerseries} holds.  
We now introduce the regularization condition:

\begin{definition}\label{def:regcond}
	A real-analytic function $f: D_0[R]\to \C$ is {\bf regularized} if $f$ is nonzero everywhere over $D_0[R]$, and the Maclaurin series of $f$ and of $f^{-1}$ have radii of convergence greater than $R$. 
\end{definition}

With the definition above, we establish the following result:

\begin{proposition}\label{prop:regcond}
When proving Theorem~\ref{thm:specialcomplex}, one can assume for free that every entry $b_i$ of the $b$-vector in system~\eqref{eq:normalform} is regularized.
\end{proposition}

\begin{proof}
We first show that the following condition can be assumed for free: every entry $b_i$ satisfies $b_i(0) \neq 0$. We do so by establishing the  fact that one can always construct another normal form $(\omega, \tilde b)$, with $\tilde b_i(0)\neq 0$ for all~$i$, such that uncontrollability of $(\omega, \tilde b)$ implies uncontrollability of $(\omega, b)$. 

To this end, we choose an arbitrary real-analytic function $b_0: D_0[R]\to\C$ with $b_0(0)\neq 0$. By concatenating $b_0$ with the row vector $b$, we obtain an augmented row vector 
$
\hat b:= [b_0, b_1, \cdots, b_m]  
$. It should be clear that $\mathcal{L}(\omega, b) \subseteq \mathcal{L}(\omega, \hat b)$. 
Next, for each $i = 1,\ldots, m$, let $\tilde b_i: D_0[R]\to \mathbb{C}$ be defined such that $\tilde b_i := b_i + b_0$ if $b_i(0)= 0$ and $\tilde b_i := b_i$ otherwise. By construction, $\tilde b_i(0)\neq 0$ for all $i = 0,\ldots, m$. Now, let $\tilde b:= [\tilde b_0,\ldots, \tilde b_m]$. 
Since each $\tilde b_i$ is a linear combination of the $b_i$ and vice versa, we have that $\mathcal{L}(\omega, \hat b) = \mathcal{L}(\omega, \tilde b)$. It then follows that $\mathcal{L}(\omega, b)\subseteq \mathcal{L}(\omega, \tilde b)$. Thus, if $(\omega, \tilde b)$ is not controllable, then neither is $(\omega, b)$.

By the above arguments, we can now assume that $b_i(0)\neq 0$ for all~$i$. 
Because $b$ is continuous and because each $b_i(0)$ is nonzero, there is a radius $R'$, with $ 0 < R' \leq R$, such that $b_i(\sigma) \neq 0$ for all $\sigma\in D_0[R']$ and for all $i = 1,\ldots, m$. Thus, $b_i$ and $b^{-1}_i$ are well defined on $D_0[R']$ and are locally represented by the corresponding Maclaurin series. By shrinking $R'$, if necessary, we can assume that $R'$ is smaller than the radii of convergence of those series. It follows that the condition given in the statement of the proposition will be satisfied if $R$ is replaced with $R'$.  
By Lemma~\ref{lem:embedding}, to show that system~\eqref{eq:normalform} is not controllable, it suffices to show that the subensemble-$D_0[R']$ is not controllable. We can thus assume that the regularization condition is satisfied without passing~\eqref{eq:normalform} to any of its subensembles. This completes the proof. 
\end{proof}

\subsection{Convergent series on annulus}\label{ssec:ringofuec}
In this subsection, we introduce the closed annulus $A_1$ as indicated earlier in the outline of proof,  
and a special class of continuous functions on $A_1$, each of which can be represented by a certain convergent series.  
To this end, let $R_1$ and $R_2$ be positive real numbers such that $0 < R_1 < R_2 < R$. Let $A[R_1, R_2]$ be a closed annulus inside $D_0[R]$: 
\begin{equation}\label{eq:annulus1}
A[R_1, R_2] := \{\sigma \in \C \mid R_1 \le |\sigma | \le R_2 \}. 
\end{equation}
For convenience, we use $A_1 := A[R_1, R_2]$ as a short notation.  
To introduce the above-mentioned continuous functions on $A_1$,  
we use polar coordinates (i.e., $\sigma = re^{\mathrm{i}\theta}$):

\begin{definition}\label{def:uec}
Let $\rho_k: [R_1, R_2] \to \C$, for $k\in \mathbb{Z}$, be continuous functions.   
	The following doubly infinite series $f: A_1\to \C$: 
	\begin{equation}\label{eq:doublyinfiniteseries}
	f(r, \theta) := \sum^\infty_{k = -\infty} \rho_k(r) e^{\mathrm{i}k\theta}
	\end{equation} 
	is {\bf uniformly and exponentially convergent}  {\normalfont(uec)} if there exists a real number $q > 1$ such that  
	$$\sum^{\infty}_{k = -\infty} \|\rho_k\|_{\mathrm{L}^\infty} q^{|k|} < \infty.$$   
	 We call $\rho_k$ the {\bf radius components} of~$f$. 
\end{definition}

Note that by the uniform limit theorem, each {\em uec} series is a continuous function. Denote by $\mathcal{K}$ the set of all {\em uec} series:
$$
\mathcal{K}:= \left \{ f \in \mathrm{C}^0(A_1, \C) \mid f \mbox{ is represented by a {\em uec} series}\right \}.
$$ 

Next, we define a set of functions $\eta_k: \mathcal{K}\to \mathrm{C}^0([R_1,R_2], \C)$, for $k\in \mathbb{Z}$, by sending an {\em uec} series $f$ to its radius components $\rho_k$. The maps $\eta_k$ are  explicitly given by: 
\begin{equation}\label{eq:defeta}
\eta_k(f)(r) := \frac{1}{2\pi}\int^\pi_{-\pi} f(r, \theta) e^{-\mathrm{i} k \theta} \mathrm{d}\theta, \quad \mbox{for } r\in [R_1,R_2].  
\end{equation}

The set $\mathrm{C}^0(A_1, \C)$ is an algebra (over $\C$) with identity: Addition and multiplication are pointwise, and the identity element is simply ${\bf 1}_{A_1}$. We have the following result:

\begin{proposition}\label{prop:closedunderadditionandmultiplication}
The set $\mathcal{K}$ is a subalgebra of $\mathrm{C}^0(A_1, \C)$ with identity. 
\end{proposition}

\begin{proof}
It should be clear that~${\bf 1}_{A_1}$ belongs to $\mathcal{K}$ 
and that $\mathcal{K}$ is a subspace of $\mathrm{C}^0(A_1, \C)$ from Definition~\ref{def:uec}. We show below that 
$\mathcal{K}$ is closed under multiplication. i.e., for any two $f_1, f_2\in \mathcal{K}$, $f_1f_2\in \mathcal{K}$.  
To proceed, we first express $f_1f_2$ using the following formal series: 
	\begin{align*}
	(f_1 f_2) (r, \theta) & = \left (\sum^\infty_{k = -\infty} \eta_k(f_1) e^{\mathrm{i}k\theta}\right) \left (\sum^\infty_{\ell = -\infty} \eta_\ell(f_2) e^{\mathrm{i}\ell\theta}\right)  \\
	& = \sum^\infty_{k = -\infty} \left [ \sum^\infty_{\ell = -\infty} \big ( \eta_{k - \ell}(f_1)  \eta_{\ell}(f_2) \big ) (r) \right ] e^{\mathrm{i}k\theta}. 
	\end{align*}
We show below that for each $k\in \mathbb{Z}$, the series $\rho_k:= \sum_{\ell = -\infty}^\infty \eta_{k - \ell}(f_1)  \eta_{\ell}(f_2)$	 is uniformly and absolutely convergent on $[R_1, R_2]$. Since $f_1, f_2\in \mathcal{K}$, by Definition~\ref{def:uec}, there exist a $p > 1$ and an $M > 0$ such that 
$\|\eta_k(f_i)\|_{\mathrm{L}^\infty} p^{|k|} < M$ for all $k\in \mathbb{Z}$ and for all $i = 1,2$.  
Then, for any $r\in [R_1,R_2]$, 
\begin{align}\label{eq:donotdependonr}
\sum_{\ell = -\infty}^\infty \left |\big (\eta_{k-\ell}(f_1)\eta_\ell(f_2) \big ) (r) \right | & \leq \sum^\infty_{\ell = -\infty} \| \eta_{k - \ell}(f_1) \|_{\mathrm{L}^\infty} \|\eta_{\ell}(f_2)\|_{\mathrm{L}^\infty}  \notag \\ 
& <   \sum^\infty_{\ell = -\infty} \frac{M^2}{p^{|k - \ell| + |\ell|}} =  \left ( |k| + \frac{p^2 + 1}{p^2 - 1} \right ) \frac{M^2}{p^{|k|}}. 
\end{align}
It now remains to show that there exists a $q > 1$ such that $\sum_{k = -\infty}^\infty \|\rho_k\|_{\mathrm{L}^\infty}q^{|k|} < \infty$.
By~\eqref{eq:donotdependonr}, we have that for any $q\in (1,p)$,   
	\begin{equation*}
	\sum_{k = -\infty}^\infty \|\rho_k\|_{\mathrm{L}^\infty} q^{|k|} <  \sum_{k = -\infty}^\infty M^2\left ( |k| + \frac{p^2 + 1}{p^2 - 1} \right ) \left ( \frac{q}{p}\right )^{|k|} < \infty.  
	\end{equation*} 
	This completes the proof.    
\end{proof}

We next introduce a set $\mathcal{P}$, obtained by restricting regularized, real-analytic functions to the annulus $A_1$. Specifically, let
\begin{equation}\label{eq:definepa1}
\mathcal{P}:= \left \{ f |_{A_1} \mid  f:D_0[R]\to \C \mbox{ is real-analytic and regularized} \right \}.  
\end{equation}  
The elements of $\mathcal{P}$ will be used in the next subsection to index a special class of  subspaces of $\mathrm{L}^2(A_1,\C)$, termed {\em featured spaces}.  
We have the following result:

\begin{proposition}\label{prop:realanalytictoK}
The set $\mathcal{P}$ defined in~\eqref{eq:definepa1} is a subset of $\mathcal{K}$. 
\end{proposition}

\begin{proof}
Let~$f: D_0[R]\to \C$ be a regularized, real-analytic function.  
Using the polar coordinates, we re-write the Maclaurin series~\eqref{eq:powerseries} of~$f$ as $
f(r, \theta) = \sum^\infty_{k = -\infty} \rho_k(r) e^{\mathrm{i}k\theta}
$, 
where $\rho_k: [0, R]\to \C$ are given by the uniformly and absolutely convergence series:  
\begin{equation}\label{eq:radiusseries}
\rho_k(r):= 
\sum^\infty_{\ell = 0} c\left (\ell + \frac{1}{2}(|k| + k),\, \ell + \frac{1}{2}(|k| - k) \right ) r^{2\ell + |k|}. 
\end{equation}	 
Next, let $\rho'_k: = \rho_k |_{[R_1,R_2]}$ and $q:= \nicefrac{R}{R_2} > 1$. Then, using~\eqref{eq:radiusseries} and the fact that $r\leq R_2$,   
	we have that 	
	\begin{align}\label{eq:fA1isinK}
	\sum^\infty_{k = -\infty} \|\rho'_k\|_{\mathrm{L}^\infty} q^{|k|} 
	& \leq \sum^\infty_{k = -\infty} \sum^\infty_{\ell = 0}\left |c\left (\ell + \nicefrac{(|k| + k)}{2},\, \ell + \nicefrac{(|k| - k)}{2} \right ) \right | R^{2\ell + |k|}_2 \left (\nicefrac{R}{R_2} \right)^{|k|} \notag \\
	& \leq \sum^\infty_{k = -\infty} \sum^\infty_{\ell = 0}\left |c\left (\ell + \nicefrac{(|k| + k)}{2},\, \ell + \nicefrac{(|k| - k)}{2} \right ) \right | R^{2\ell + |k|} \notag \\
	& = \sum^\infty_{k = 0} \sum^\infty_{\ell= 0} |c(k,\ell)| R^{k + \ell}.
	\end{align}
	Since $f$ is regularized, the radius of convergence of its Maclaurin series is greater than~$R$ and, hence, the last expression~\eqref{eq:fA1isinK} is bounded above.      
\end{proof}

\subsection{Finite intersection property}\label{ssec:finiteintersection}
In this subsection, we first introduce and characterize a special class of Hilbert subspaces of $\mathrm{L}^2(A_1, \C)$, indexed by elements in~$\mathcal{P}$. We next formulate a theorem, Theorem~\ref{thm:intersection}, which states that intersections of finitely, but arbitrarily, many these subspaces are always nontrivial. 
Theorem~\ref{thm:specialcomplex} will then follow as an immediate consequence of Theorem~\ref{thm:intersection}.

Recall that $\omega_{A_1}$ is the identity function on $A_1$. For ease of notation, we will omit its subindex in the sequel. Let $\mathcal{P}$ be given as in~\eqref{eq:definepa1}. For any $g\in \mathcal{P}$, let  
	\begin{equation}\label{eq:featured}
		\mathcal{K}_g:=  \mathcal{K} \cap \mathcal{L}^\perp(\omega, g) = \left \{f\in \mathcal{K} \mid \langle f, \omega^kg \rangle_{A_1} = 0, \mbox{ for all } k\ge 0 \right \}, 
	\end{equation} 
	Note that the constant function $\mathbf{1}$ belongs to $\mathcal{P}$ (its subindex $A_1$ has been omitted).    
We characterize below the subspaces~$\mathcal{K}_g$:

\begin{proposition}\label{prop:subspacekgeneral}
The following two items hold:
\begin{enumerate}
\item Let $\eta_k$ be defined as in~\eqref{eq:defeta}. The set $\mathcal{K}_{\mathbf{1}}$ comprises all $f\in \mathcal{K}$ such that 
\begin{equation}\label{eq:conditionfork1}
\int^{R_2}_{R_1} \eta_{k}(f)(r)r^{k + 1} \mathrm{d}r = 0, \quad \mbox{for all } k \ge 0.
\end{equation}  
\item Let $g$ be an arbitrary element in $\mathcal{P}$. Then, an $f\in \mathcal{K}$ belongs to $\mathcal{K}_g$ if and only if there exists an $f' \in \mathcal{K}_{\mathbf{1}}$ such that $f = f' \bar g^{-1}$. 
\end{enumerate}	
\end{proposition}

\begin{proof} 
We first establish item 1. 
Using polar coordinates, we have that  
\begin{equation}\label{eq:fidk}
\langle \omega^k, f \rangle_{A_1} = \int^\pi_{-\pi}\int^{R_2}_{R_1} \sum^{\infty}_{\ell = -\infty}  \eta_\ell(f)(r) r^{k+1} e^{\mathrm{i}(\ell-k)\theta}\mathrm{d}r\mathrm{d}\theta.   
\end{equation}
Since $f\in \mathcal{K}$, there exists an $M>0$ such that $\sum_{\ell = -\infty}^\infty \|\eta_\ell(f)\|_{\mathrm{L}^\infty} < M$ and, hence, 
$$
\sum_{\ell = -\infty}^\infty|\eta_\ell(f)(r) r^{k+1} e^{\mathrm{i}(\ell - k)\theta} |  < R_2^{k+1} M, \quad \mbox{for all } r\in [R_1,R_2].  
$$
By dominated convergence theorem, we can switch the order of integrals and sum in~\eqref{eq:fidk} and obtain that  
$$
\langle \omega^k, f \rangle_{A_1} =\sum^{\infty}_{\ell = -\infty} \int^{R_2}_{R_1} \int^\pi_{-\pi} \eta_\ell(f)(r) r^{k+1} e^{\mathrm{i}(\ell - k)\theta}\mathrm{d}\theta\mathrm{d}r = 2\pi \int^{R_2}_{R_1}  \eta_{k}(f)(r)r^{k + 1} \mathrm{d}r.  
$$
Thus,  $\langle \omega^k, f \rangle_{A_1} = 0$ if and only if~\eqref{eq:conditionfork1} holds. This establishes item 1.

We next establish item 2. 
First, for any $f\in \mathcal{K}_g$, we let $f':= f\bar g$ and show that $f'\in \mathcal{K}_{\mathbf{1}}$.  
Since $g\in \mathcal{P}$, its complex conjugation $\bar g$ also belongs to $\mathcal{P}$ and, hence,   
to $\mathcal{K}$ by Proposition~\ref{prop:realanalytictoK}. 
By Proposition~\ref{prop:closedunderadditionandmultiplication}, $\mathcal{K}$ is closed under multiplication, so  
$f' \in \mathcal{K}$. Moreover,   
\begin{equation}\label{eq:f'fbarg}
\langle f', \omega^k \rangle_{A_1} = \langle f \bar g, \omega^k \rangle_{A_1} = \langle f, \omega^k g  \rangle_{A_1} = 0,
\end{equation}     
which implies that $f'\in \mathcal{K}_{\mathbf{1}}$. Conversely, for any $f'\in \mathcal{K}_{\bf 1}$, let $f := f' \bar g^{-1}$. Since $g\in \mathcal{P}$, $g$ is regularized and, hence, $g^{-1}$ belongs to $\mathcal{P}$. By the same arguments, $f\in \mathcal{K}$. 
Using again~\eqref{eq:f'fbarg}, we conclude that $f\in \mathcal{K}_g$. This establishes item 2.  
\end{proof}

It should be clear that for any $g\in \mathcal{P}$, $\mathcal{K}_g$ is nontrivial, i.e., it contains nonzero elements.  
Indeed, for the special case $g = {\bf 1}$, the functions $\overline{\omega}^k$, for $k\geq 1$, belong to $\mathcal{K}_{\bf 1}$.  
Then, by item 2 of Proposition~\ref{prop:subspacekgeneral}, $\mathcal{K}_g$ is nontrivial for all $g\in \mathcal{P}$. 
The following result shows that the intersection of finitely, but arbitrarily many, $\mathcal{K}_g$ is also nontrivial:

\begin{theorem}\label{thm:intersection}
Let $A_1 = A[R_1, R_2]$ be the closed annulus in $D_0[R]$ defined in~\eqref{eq:annulus1} and $\mathcal{P}$ be defined in~\eqref{eq:definepa1}. For each $g\in \mathcal{P}$, let $\mathcal{K}_g$ be defined in~\eqref{eq:featured}.  
Suppose that $R_1 R > R_2^2$; 
then, for any finite set $\{g_1,\ldots, g_m\}$ out of $\mathcal{P}$,  $\cap^m_{i = 1} \mathcal{K}_{g_i}$ is nontrivial. 
\end{theorem}

We call the property described in the above theorem the {\em finite intersection property}.  
Theorem~\ref{thm:specialcomplex} then follows as an immediate consequence of Theorem~\ref{thm:intersection}:

{\em Proof of Theorem~\ref{thm:specialcomplex}.}
	Let $b_1,\ldots, b_m$ be the $m$ entries of the $b$-vector of system~\eqref{eq:normalform}. By Proposition~\ref{prop:regcond}, we can assume for free that each $b_i$ is regularized. 
	Let $g_i:= b_i |_{A_1}$ for all $i = 1,\ldots, m$; then, $g_i\in \mathcal{P}$. By Theorem~\ref{thm:intersection}, there is a nonzero $f_0\in \cap^m_{i = 1} \mathcal{K}_{g_i}$. By definition~\eqref{eq:featured} of $\mathcal{K}_{g_i}$, $f_0$ is perpendicular to $\mathcal{L}(\omega, g)$, where $g:=[g_1,\ldots, g_m]$. 
	We then extend $f_0$ to a nonzero function $\tilde f_0\in \mathrm{L}^2(D_0[R],\C)$ by setting $\tilde f(\sigma) := 0$ for all $\sigma\in D_0[R]\backslash A_1$. By construction, $\tilde f_0$ is perpendicular to $\mathcal{L}(\omega, b)$. We then conclude from Lemma~\ref{lem:controllablesubspace} that the pair $(\omega, b)$ is not $\mathrm{L}^2$-controllable.   
\endproof

The remainder of the section is devoted to the proof of Theorem~\ref{thm:intersection}. 
We will show that there exist a nonzero $f_0 \in \mathcal{K}$ and $m$  functions $f_1,\ldots, f_m \in \mathcal{K}_{\bf 1}$ such that 
\begin{equation}\label{eq:condition1h}
f_0 \bar g_i = f_i, \quad \mbox{for all } i = 1,\ldots, m.
\end{equation} 
Note that if~\eqref{eq:condition1h} holds, then $f_0 = f_i \bar g_i^{-1}$ for all $i = 1,\ldots, m$. By Proposition~\ref{prop:subspacekgeneral}, 
$f_0 \in \cap^m_{i = 1}\mathcal{K}_{g_i}$, i.e., Theorem~\ref{thm:intersection} is established.

\subsection{Null condition}\label{ssec:core}  
    
In this subsection, we establish a necessary and sufficient condition, termed {\em null condition}, for two functions $f,g\in \mathcal{K}$ to satisfy $f g\in \mathcal{K}_{\mathbf{1}}$. 
Recall that $R_1$ and $R_2$ are the inner- and outer-radii of the annulus $A_1$. For convenience, let $s_1:=R^2_1$ and $s_2:=R^2_2$. 
Using the functions $\eta_k$ defined in~\eqref{eq:defeta}, 
we introduce another set of functions $\xi_k: \mathcal{K}\to \mathrm{C}^0([s_1,s_2], \C)$, for $k\in \mathbb{Z}$. For any $f\in \mathcal{K}$, let 
$\xi_k(f): [s_1, s_2]\to \C$ be defined as follows: 
\begin{equation}\label{eq:defxik}
\xi_k(f)(s) := \eta_k(f)(\sqrt{s}) s^{\frac{k}{2}} = \frac{1}{2\pi}  \int_{-\pi}^{\pi} f(\sqrt{s}, \theta) e^{-\mathrm{i}k\theta} \mathrm{d}\theta s^{\frac{k}{2}}.
\end{equation}
To introduce the null condition, we first have the following result:

\begin{proposition}\label{prop:goalupdate1}
	Let $\xi_k$ be defined as in~\eqref{eq:defxik}. Then, for any $f, g\in \mathcal{K}$,  
	\begin{equation}\label{eq:interchangesumandintegral}
	\int^{R_2}_{R_1}\eta_k(fg)(r) r^{k+1} \mathrm{d}r = \frac{1}{2}\sum^\infty_{\ell = -\infty}\left \langle \bar \xi_{\ell}(f), \, \xi_{k-\ell}(g) \right \rangle, \quad \mbox{for all } k\in \mathbb{Z}.  
	\end{equation} 
	
\end{proposition}

\begin{proof}  
First, note that $\eta_k(fg) =    
 \sum^\infty_{\ell = -\infty} \eta_{\ell}(f)  \eta_{k-\ell}(g)$; the series is uniformly and absolutely convergent as shown in the proof of Proposition~\ref{prop:closedunderadditionandmultiplication}. 
 It then follows that  
 \begin{align}\label{eq:fourthstepdh1}
 \int^{R_2}_{R_1}\eta_k(fg)(r) r^{k+1} \mathrm{d}r & = \int^{R_2}_{R_1}  \sum^\infty_{\ell = -\infty} \big (\eta_{\ell}(f)  \eta_{k-\ell}(g) \big ) (r) r^{k+1} \mathrm{d}r \notag \\
 & = \sum^\infty_{\ell = -\infty} \int^{R_2}_{R_1} \big (\eta_{\ell}(f) \eta_{k-\ell}(g) \big ) (r) r^{k + 1} \mathrm{d}r,
\end{align}
where the last equality follows from the dominated convergence theorem. 
From~\eqref{eq:defxik}, we have that 
$$\eta_\ell(f)(r) = \xi_\ell(f)(r^2) r^{-\ell} \quad \mbox{and} \quad \eta_{k-\ell}(g)(r) = \xi_{k-\ell}(g)(r^2) r^{\ell-k}.$$ 
Using the above two expressions and changing variable $s:=r^2$, we obtain that
\begin{equation}\label{eq:laststepforgoalupdate1}
\int^{R_2}_{R_1} \big (\eta_{\ell}(f) \eta_{k-\ell}(g) \big ) (r) r^{k + 1} \mathrm{d}r = \frac{1}{2} \int_{s_1}^{s_2} \big (\xi_{\ell}(f) \xi_{k-\ell}(g) \big )(s) \mathrm{d}s. 
\end{equation} 
By~\eqref{eq:fourthstepdh1} and~\eqref{eq:laststepforgoalupdate1}, we conclude that~\eqref{eq:interchangesumandintegral} holds. 
\end{proof}

The next result is then an immediate consequence of Propositions~\ref{prop:subspacekgeneral} and~\ref{prop:goalupdate1}:  

\begin{corollary}\label{cor:nullcondition}
Let $f, g\in \mathcal{K}$. Then, $fg$ belongs to $\mathcal{K}_{\bf 1}$ if and only if  
	\begin{equation}\label{eq:gxfc}
	\sum^\infty_{\ell = -\infty}\left \langle \bar \xi_{\ell}(f), \, \xi_{k-\ell}(g) \right \rangle = 0, \quad \mbox{for all } k \ge 0.
	\end{equation}
\end{corollary}
We call~\eqref{eq:gxfc} the {\em null condition}.

Now, to establish Theorem~\ref{thm:intersection}, it suffices to show that given any finite subset $\{g_1,\ldots, g_m\}$ of $\mathcal{P}$, there exists a nonzero $f_0\in \mathcal{K}$ such that~\eqref{eq:gxfc} is satisfied, with $f$ replaced by $f_0$ and $g$ replaced by $g_i$, for all $i = 1,\ldots,m$. We address this existence problem by first introducing a set of Laurent series induced by the $g_i$ (Subsection~\S\ref{ssec:laurent}) and, then, providing  a nontrivial solution (which will be used to construct~$f_0$) to a homogeneous linear equation over the ring of these Laurent series (Subsection~\S\ref{ssec:proof}).

\subsection{Connections with Laurent series}\label{ssec:laurent} 
Let $A_2 := A[\nicefrac{R^2_2}{R}, R]$ be a closed annulus in $D_0[R]$, with inner- and outer-radii being $\nicefrac{R^2_2}{R}$ and $R$, respectively. Define  
\begin{equation}\label{eq:defcalH}
\mathcal{H}:= \left \{ h: A_2 \to \C \mid h \mbox{ is holomorphic on } A_2 \right \}.
\end{equation} 
Let $\mathcal{P}$ be defined as in~\eqref{eq:definepa1}. 
We construct below a set of functions $\phi_n: \mathcal{P}\to \mathcal{H}$ for all integers $n\ge 0$. 
To this end, let $\{p_n\}^\infty_{n = 0}$ be an orthonormal basis of ${\rm L}^2([s_1, s_2], \C)$.   
We will assume that every $p_n$ is a polynomial of degree~$n$ with real coefficients. Such a basis can be obtained, for example, from the monomials $\{s^n\}_{n\ge 0}$ by applying the Gram-Schmidt process. 
Now, for each $n\ge 0$ and for any given $g\in \mathcal{P}$, we define a Laurent series $\phi_n(g)$ as follows: 
\begin{equation}\label{eq:deflaurentseries}
\phi_n(g)(z):= \sum_{k=-\infty}^\infty \langle p_n,\, \xi_{-k}(g)\rangle z^k,  
\end{equation}
where functions $\xi_k$ are defined in~\eqref{eq:defxik}. 
We now have the following fact:

\begin{proposition}\label{prop:decayalpha}
	For any $g\in \mathcal{P}$ and for any $n \ge 0$, $\phi_n(g)\in \mathcal{H}$. 
\end{proposition}

\begin{proof}
Since $g\in \mathcal{P}$, there is a regularized, real-analytic function $g': D_0[R]\to \C$ such that $g = g'|_{A_1}$. 
We express $g'$ using its Maclaurin series as follows: 
\begin{equation}\label{eq:mlseriesofg'}
g'(\sigma) = \sum_{k = 0}^\infty \sum_{\ell = 0}^\infty c'(k,\ell) \sigma^k\bar \sigma^\ell,
\end{equation} 
where $c'(k,\ell)\in\C$. By Definition~\ref{def:regcond}, the radius of convergence of the above series is greater than $R$. Thus, there exists an $\epsilon > 0$ such that 
\begin{equation}\label{eq:boundforsumofc'kl}
\sum_{k = 0}^\infty\sum_{\ell = 0}^\infty |c'(k,\ell)| (R + \epsilon)^{k +\ell} < \infty. 
\end{equation}
We show below that for the given $\epsilon$, the Laurent series $\phi_{n}(g)$, for any $n\ge 0$, converges uniformly and absolutely on the closed annulus $A'_2:= A[\nicefrac{R^2_2}{(R + \epsilon)}, R + \epsilon]$, which contains $A_2$ as a proper subset.  

First, note that for any $z\in A'_2$ and for any $k\in \mathbb{Z}$, 
\begin{equation}\label{eq:zabsk}
|z|^{k} \leq \max \left \{ (R + \epsilon)^k, \frac{R^{2k}_2}{(R + \epsilon)^k} \right \} = \frac{(R + \epsilon)^{|k|}}{R_2^{|k| - k}}.  
\end{equation}  
Also, note that 
\begin{equation}\label{eq:boundforinnerproduct}
|\langle p_n,\, \xi_{-k}(g)\rangle|\leq (s_2 - s_1) \|p_n\|_{\mathrm{L}^\infty} \|\xi_{-k}(g)\|_{\mathrm{L}^\infty}. 
\end{equation} 
Since $p_n$ is a polynomial and $\xi_{-k}(g)$ is continuous (both  are defined over $[s_1,s_2]$), we have that $\|p_n\|_{\mathrm{L}^\infty}$ and $\|\xi_{-k}(g)\|_{\mathrm{L}^\infty}$ are finite. Then, using~\eqref{eq:zabsk} and~\eqref{eq:boundforinnerproduct}, we obtain that for any $z\in A_2'$,    
\begin{equation}\label{eq:betaiknzk}
\sum^\infty_{k = -\infty} |\langle p_n,\, \xi_{-k}(g)\rangle| |z|^{k} \le  (s_2 - s_1) \|p_n \|_{\mathrm{L}^\infty} \sum^\infty_{k = -\infty} \|\xi_{-k}(g)\|_{\mathrm{L}^\infty} \frac{(R + \epsilon)^{|k|}}{R_2^{|k| - k}}.
\end{equation}

We now show that the infinite sum on the right hand side of~\eqref{eq:betaiknzk} is bounded. 
To proceed, we first obtain an upper bound for $\|\xi_{-k}(g)\|_{\mathrm{L}^\infty}$. From~\eqref{eq:defxik}, we have that $\xi_{-k}(g)(s) = \eta_{-k}(g)(\sqrt{s})s^{-\nicefrac{k}{2}}$. 
We can express $\eta_{-k}(g)$ using the coefficients $c'(\cdot,\cdot)$ in the Maclaurin series~\eqref{eq:mlseriesofg'} of~$g'$ as follows:  
$$
\eta_{-k}(g)(r) = \sum^\infty_{\ell = 0} c'\left (\ell + \nicefrac{(|k| - k)}{2},\, \ell + \nicefrac{(|k| + k)}{2} \right ) r^{2\ell + |k|}. 
$$    
It then follows that 
$$
\xi_{-k}(g)(s) = \sum^\infty_{\ell = 0} c'\left (\ell + \nicefrac{(|k| - k)}{2},\, \ell + \nicefrac{(|k| + k)}{2} \right ) s^{\ell + \frac{1}{2}(|k| - k)}. 
$$
Because $s\in [R^2_1, R^2_2]$ and $R_2 < R$, we obtain that
\begin{equation}\label{eq:infinitynormofxikg}
\|\xi_{-k}(g)\|_{\mathrm{L}^\infty} \leq \sum^\infty_{\ell = 0} \left | c'\left (\ell + \nicefrac{(|k| - k)}{2},\, \ell + \nicefrac{(|k| + k)}{2} \right ) \right | (R+\epsilon)^{2\ell} R_2^{|k|-k}. 
\end{equation}

With~\eqref{eq:infinitynormofxikg}, we can now provide an upper bound for the infinite sum on the right hand side of~\eqref{eq:betaiknzk}: 
 \begin{align*}
 \sum^\infty_{k = -\infty} \|\xi_{-k}(g)\|_{\mathrm{L}^\infty} \frac{(R + \epsilon)^{|k|}}{R_2^{|k| - k}} & \leq \sum_{k = -\infty}^\infty\sum^\infty_{\ell = 0} \left | c'\left (\ell + \nicefrac{(|k| - k)}{2},\, \ell + \nicefrac{(|k| + k)}{2} \right ) \right | (R+\epsilon)^{2\ell + |k|} \\
 & = \sum_{k = 0}^\infty\sum_{\ell = 0}^\infty |c'(k,\ell)| (R + \epsilon)^{k +\ell} < \infty, 
 \end{align*}
where the last inequality follows from~\eqref{eq:boundforsumofc'kl}. This completes the proof. 
\end{proof}

\subsection{Proof of Theorem~\ref{thm:intersection}}\label{ssec:proof} 
Let $\{g_1,\ldots, g_m\}$ be an arbitrary finite subset of~$\mathcal{P}$. 
We will first construct a nonzero element $f_0\in \mathcal{K}$ and, then, show that  $f_0\bar g_i\in \mathcal{K}_{\mathbf{1}}$ for all $i = 1,\ldots, m$. In the sequel, we will assume that $RR_1> R_2^2$, which is the hypothesis of Theorem~\ref{thm:intersection}. This hypothesis will be instrumental in showing that the function $f_0$ constructed below belongs to $\mathcal{K}$.

\subsubsection{Construction of $f_0$} Let the sets $\mathcal{P}$ and $\mathcal{H}$, and the maps $\phi_n:\mathcal{P}\to \mathcal{H}$, for $n\geq 0$, be defined as in~\eqref{eq:definepa1},~\eqref{eq:defcalH}, and~\eqref{eq:deflaurentseries}, respectively.  
We first have the following result:

\begin{lemma}\label{lem:nonzeropsi}
There exist $\psi_0,\ldots,\psi_m\in \mathcal{H}$, with at least one nonzero $\psi_i$, such that the following holds: 
\begin{equation}\label{eq:nonzeropsi}
	\sum_{n=0}^m \phi_n(\bar g_i) \psi_n = 0, \quad \mbox{for all } i = 1,\ldots, m.
\end{equation}
\end{lemma}
 
\begin{proof}
The $(m+1)$ linear homogeneous equations in~\eqref{eq:nonzeropsi} form an underdetermined system, with $m$ unknowns $\psi_1,\ldots,\psi_n$, over the ring $\mathcal{H}$. Since $\mathcal{H}$ is an integral domain, such a system has a nonzero solution. 
\end{proof}

Let the Laurent expansions of $\psi_n$, for $n = 0,\ldots,m$, be given by 
\begin{equation}\label{eq:introducealpha}
	\psi_n(z) = \sum_{k = -\infty}^\infty\alpha_{n,k} z^k.
\end{equation}
Using the coefficients $\alpha_{n,k}$ in~\eqref{eq:introducealpha}, we define functions $\rho_{0,k}:[R_1, R_2]\to\C$, for $k\in \mathbb{Z}$, as follows:
\begin{equation}\label{def:eqrhokforf0}
	\rho_{0,k}(r):= \sum^m_{n = 0} \alpha_{n,-k} p_n(r^2) r^{-k}, 
\end{equation} 
where we recall that each $p_n$ is a polynomial of degree $n$ with real coefficient, and that $\{p_n\}_{n\ge 0}$ is an orthonormal basis of $\mathrm{L}^2([s_1,s_2], \C)$. 
With $\rho_{0,k}$ defined in~\eqref{def:eqrhokforf0}, we set 
\begin{equation}\label{eq:deff0}
	f_0(r,\theta):= \sum_{k = -\infty}^\infty \rho_{0,k}(r) e^{\mathrm{i}k\theta}.
\end{equation}
We now establish the following result:

\begin{proposition}\label{prop:constructionoff0}
The function $f_0$ defined by~\eqref{eq:introducealpha},~\eqref{def:eqrhokforf0}, and~\eqref{eq:deff0}  is a nonzero element in $\mathcal{K}$.  
\end{proposition}

\begin{proof}
We first show that $f_0\in \mathcal{K}$ and, then, show that $f_0$ is nonzero. 
Since $r\in [R_1, R_2]$, we obtain from~\eqref{def:eqrhokforf0} that 
$$
\|\rho_{0,k}\|_{\mathrm{L}^\infty} \leq \left ( \max^{m}_{n = 0}\|p_n\|_{\mathrm{L}^\infty} \right ) \left (\max^2_{i=1} R_i^{-k} \right )\sum_{n = 0}^m |\alpha_{n,-k}|, \quad \mbox{for all } k\in \mathbb{Z}.  
$$ 
It follows that for any $q>1$, 
\begin{multline}\label{eq:boundedaboveforf0}
\sum^\infty_{k = -\infty} \|\rho_{0,k} \|_{\mathrm{L}^\infty} q^{|k|}  \leq 
\left ( \max^{m}_{n = 0}\|p_n\|_{\mathrm{L}^\infty} \right ) \sum^{m}_{n = 0}\sum^\infty_{k = 0} \left ( |\alpha_{n,k}| (q R_2)^k + |\alpha_{n,-k}| \left (\nicefrac{q}{R_1} \right )^{k}  \right ).
\end{multline}
We exhibit below a $q>1$ such that the right hand side of~\eqref{eq:boundedaboveforf0} is bounded.

Since each $\psi_n(z) =  \sum^\infty_{k = -\infty} \alpha_{n,k} z^{k}$, for $n = 0,\ldots, m$, is holomorphic on $A_2 = A[\nicefrac{R^2_2}{R}, R]$, it is absolutely convergent on the inner- and outer-circles of the annulus. 
	It follows that for all $n = 0  ,\ldots, m$ and for all $k\ge 0$,   
\begin{equation}\label{eq:boundednessofalphasequence}
\sum^\infty_{k = 0} |\alpha_{n,k}| \left (\nicefrac{R}{R^2_2} \right )^k < \infty  \quad \mbox{and} \quad
\sum^\infty_{k = 0} |\alpha_{n,-k}| R^k < \infty.\\
\end{equation} 
Now, set $q := \nicefrac{R R_1}{R^2_2}$. By the hypothesis, $RR_1 > R_2^2$, so $q > 1$. It follows that   
\begin{align*}
\sum^\infty_{k = 0} \left ( |\alpha_{n,k}| (q R_2)^k + |\alpha_{n,-k}| \left (\nicefrac{q}{R_1} \right )^{k}  \right ) & =  
\sum^\infty_{k = 0} \left ( |\alpha_{n,k}| \left (\nicefrac{RR_1}{R_2} \right )^k + |\alpha_{n,-k}| \left (\nicefrac{R}{R_2^2} \right )^k \right ) \\ 
& \leq \sum^\infty_{k = 0} \left ( |\alpha_{n,k}| R^k + |\alpha_{n,-k}| \left (\nicefrac{R}{R_2^2} \right )^k \right ) < \infty,
\end{align*}
where the first inequality follows from the fact that $R_1 < R_2$ (and, hence, $\nicefrac{RR_1}{R_2} < R$) and the last inequality follows from~\eqref{eq:boundednessofalphasequence}. Because the above holds for all $n = 0,\ldots, m$, the right hand side of~\eqref{eq:boundedaboveforf0} is bounded above. 

It now remains to show that $f_0$ is nonzero. Note that 
$$
\|f_0\|^2_{\mathrm{L}^2} = 2\pi\sum_{k = -\infty}^\infty \int_{R_1}^{R_2}|\rho_{0,k}(r)|^2 r\mathrm{d}r.
$$ 
Thus, it suffices to show that there exists at least one nonzero $\rho_{0,k}$ for some~$k\in \mathbb{Z}$. By Lemma~\ref{lem:nonzeropsi}, there exists an $\psi_{n'}$, for some $n' \in \{0,\ldots, m\}$, such that $\psi_{n'} \ne 0$. It follows from the Laurent expansion~\eqref{eq:introducealpha} that there exists an $\alpha_{n',k'}$, for some $k'\in \mathbb{Z}$, such that $\alpha_{n',k'}\neq 0$. We claim that $\rho_{0,-k'}\neq 0$. To see this, note that
$$
\int_{R_1}^{R_2} \rho_{0,-k'}(r) p_{n'}(r^2)r^{1-k'}\mathrm{d}r = \sum_{n = 0}^m\alpha_{n,k'} \langle p_n, p_{n'}\rangle = \alpha_{n',k'} \neq 0.  
$$
This completes the proof. 
\end{proof}

\subsubsection{Proof that $f_0\bar g_i\in \mathcal{K}_{\mathbf{1}}$} 
By Corollary~\ref{cor:nullcondition}, to show that $f_0\bar g_i\in \mathcal{K}_{\bf 1}$, it suffices to establish the following result: 

\begin{proposition}\label{prop:connectiontoorthognality}
	Let $f_0\in \mathcal{K}$ be defined by~\eqref{eq:introducealpha},~\eqref{def:eqrhokforf0}, and~\eqref{eq:deff0}, and let $\xi_k$ be defined as in~\eqref{eq:defxik}.  Then, for all $i = 1,\ldots, m$, 
	\begin{equation}\label{eq:gxfc1}
	\sum_{\ell = - \infty}^\infty\langle \bar \xi_\ell(f_0), \xi_{k - \ell}(\bar g_i) \rangle = 0, \quad \mbox{for all } k\ge 0.    
	\end{equation} 		
\end{proposition}

\begin{proof}
It should be clear from Proposition~\ref{prop:constructionoff0} that the radius components of $f_0$ are given by $\eta_\ell(f_0) = \rho_{0,\ell}$, where $\rho_{0,\ell}$ for $\ell\in \mathbb{Z}$ are defined in~\eqref{def:eqrhokforf0}. Since $\xi_\ell(f_0)(s) = \eta_\ell(f_0)(\sqrt{s})s^{\ell/2}$, we obtain that 
\begin{equation}\label{eq:xiellf0}
\xi_\ell(f_0)(s) = \sum_{n = 0}^m \alpha_{n,-\ell}p_n(s).
\end{equation}
Next, for convenience,  we introduce for each $\bar g_i$ the following complex numbers: 
\begin{equation}\label{eq:introducebetas}
\beta_{n,k}(\bar g_i):= \langle p_n, \xi_{-k}(\bar g_i)\rangle \quad \mbox{for all } n\ge 0 \mbox{ and for all } k\in \mathbb{Z}.
\end{equation} 
By~\eqref{eq:deflaurentseries}, these numbers are the coefficients in the Laurent expansions of $\phi_n(\bar g_i)$, i.e., 
\begin{equation}\label{eq:laurentexpansionofgi}
\phi_n(\bar g_i) = \sum_{k = -\infty}^\infty\beta_{n,k}(\bar g_i)z^k.
\end{equation} 
Since $\{p_n\}_{n\ge 0}$ is an orthonormal basis of $\mathrm{L}^2([s_1,s_2],\C)$, 
we can thus express $\xi_{k-\ell}(\bar g_i)$ in the $\mathrm{L}^2$-sense as follows:
\begin{equation}\label{eq:xikellbargi}
\xi_{k-\ell}(\bar g_i) = \sum_{n = 0}^\infty \beta_{n,\ell-k}(\bar g_i)p_n.
\end{equation}
Using the two expressions~\eqref{eq:xiellf0} and~\eqref{eq:xikellbargi} and, again, the fact that $\{p_n\}_{n\ge 0}$ is an orthonormal basis, we obtain that
\begin{equation}\label{eq:alphabetais0}
\langle \bar \xi_\ell(f_0), \xi_{k - \ell}(\bar g_i) \rangle = \sum_{n = 0}^m\alpha_{n,-\ell} \beta_{n,\ell-k}(\bar g_i).
\end{equation}
Thus,~\eqref{eq:gxfc1} holds if and only if 
\begin{equation}\label{eq:gxfc2}
\sum_{\ell = -\infty}^\infty \sum_{n = 0}^m \alpha_{n,-\ell} \beta_{n,\ell-k}(\bar g_i) = 0, \quad \mbox{for all } k\ge 0.
\end{equation}

Now, consider the Laurent expansions of $\phi_n(\bar g_i)\psi_n $, for $n = 0,\ldots, m$:
\begin{equation}
	\left (\phi_n(\bar g_i)\psi_n  \right )(z) = \sum_{k = -\infty}^\infty \gamma_{n,k}(\bar g_i) z^k.
\end{equation}
The Laurent expansions of  $\psi_{n}$ and of $\phi_{n}(\bar g_i)$ are given by~\eqref{eq:introducealpha} and~\eqref{eq:laurentexpansionofgi}, respectively. It follows that the coefficients~$\gamma_{n,k}(\bar g_i)$ are given by
$$
\gamma_{n,k}(\bar g_i) = \sum_{\ell = -\infty}^\infty\alpha_{n,-\ell} \beta_{n,\ell+k}(\bar g_i), \quad \mbox{for all } k\in \mathbb{Z} \mbox{ and for all } n = 0,\ldots, m.  
$$ 
Thus, to establish~\eqref{eq:gxfc2}, it suffices to show that 
$\sum_{n = 0}^m \gamma_{n,k}(\bar g_i) = 0$ for all $k \leq 0$. 
But, this directly follows from Lemma~\ref{lem:nonzeropsi}; indeed, by~\eqref{eq:nonzeropsi}, we obtain that 
$$
\left (\sum_{n = 0}^m\phi_n(\bar g_i)\psi_n \right )(z) = \sum_{k = -\infty}^\infty \left ( \sum_{n = 0}^m \gamma_{n,k}(\bar g_i) \right ) z^k = 0.
$$
If a holomorphic function is identically zero, then all of its coefficients in the associated Laurent expansion vanish. This completes the proof.     
\end{proof}

\section{Reductions and Translations to Normal Forms}\label{sec:proofofmain}
In this section, we assume that $\Sigma$ contains an open set $U$ in $\R^d$, with $d\geq 2$, and prove Theorem~\ref{thm:main}. Following the results in Section~\S\ref{sec:preliminaryresult}, we can assume, without loss of generality, that $A:\Sigma\to \C^{n\times n}$ and $B:\Sigma\to \C^{n\times m}$ are real-analytic, matrix-valued functions. 
We will show that the following complex linear ensemble system: 
\begin{equation}\label{eq:modelrep}
\dot x(t, \sigma) = A(\sigma) x(t, \sigma) + B(\sigma)u(t), \quad \mbox{for all } \sigma\in \Sigma,
\end{equation}
is not $\mathrm{L}^2$-controllable.  
The proof relies on the fact that any such system~\eqref{eq:modelrep} can be turned into a normal form~\eqref{eq:normalform} after a sequence of reductions and translations.

\subsection{Reduction on state space}\label{ssec:reductiononstatespace}

In this subsection, we perform reduction on state spaces of individual systems. The process takes two steps: In the first step, we find a closed, $d$-dimensional ball $\Sigma'$ as a subset of $\Sigma$ such that a branch of eigenvalues of the $A$-matrix and its corresponding eigenspaces are real-analytic over $\Sigma'$. Thanks to Lemma~\ref{lem:embedding}, to show that the original system is not controllable, we only need to show that the subensemble-$\Sigma'$ is not. In the second step, we translate the subensemble to a system whose $A$-matrix is block upper triangular via a similarity transformation. We then make use of such a structure and iteratively reduce the dimensions of individual systems. The iteration terminates in finite steps and we end up with a real-analytic {\em scalar} ensemble over the closed ball $\Sigma'$.          
 
\subsubsection{Local real-analyticity of eigenvalues and eigenspaces}    
Let $\eig(\sigma)$ be the set of eigenvalues of $A(\sigma)$. For a given $\lambda\in \eig(\sigma)$, let $m_a(\lambda)$ be the algebraic multiplicity of $\lambda$. 
Let $U$ be the open set in $\R^d$ inside $\Sigma$. 
We then let  
$$\lambda_a \in \argmin \{m_a(\lambda) \mid \lambda \in \eig(\sigma) \mbox{ and } \sigma\in U \}.$$ 
Let $k_a:= m_a(\lambda_a)$ and $\sigma_a\in U$ be such that $\lambda_a\in \eig(\sigma_a)$. Let $U_a$ be an open neighborhood of $\sigma_a$ inside $U$ and $\lambda: U_a \to \C$ be a continuous function such that $\lambda(\sigma)\in \eig(\sigma)$ for all $\sigma\in U_a$ and $\lambda(\sigma_a) = \lambda_a$  (continuity of $\lambda$ can be established via the use of Rouch\'e's Theorem). 
Since $m_a(\lambda(\sigma))$ is locally nonincreasing in~$\sigma$ and attains the minimum value at $\sigma_a$, one can shrink $U_a$, if necessary, so that  
$m_a(\lambda(\sigma)) = k_a$ for all $\sigma\in U_a$. 
We have the following result:  

\begin{lemma}\label{lem:branchofeigenvalues}
The function $\lambda: U_a \to \C$ is real-analytic. 	
\end{lemma}
 
\begin{proof}
	Consider the function $h: U_a\times \C\to \C$ defined as follows: $$h(\sigma,t):= \frac{\partial^{k_a - 1}}{\partial t^{k_a -1}} \det(tI_n - A(\sigma)).$$  
	It should be clear that $h(\sigma, t)$ is a polynomial in $t$ for any fixed $\sigma$. By construction of $k_a$ and $U_a$, this polynomial has a simple root $\lambda(\sigma)$ and, hence,    
	$\frac{\partial }{\partial t} h(\sigma, \lambda(\sigma))\neq 0$.   
	The lemma then follows from the analytic implicit function theorem.
\end{proof}
 
We fix the branch of eigenvalues $\lambda: U_a \to \C$ constructed above.     
Let $m_g(\lambda(\sigma))$ be the geometric multiplicity of~$\lambda(\sigma)$. 
Let 
$$\sigma_g\in \argmin\{m_g(\lambda(\sigma))\mid \sigma\in U_a \} \quad \mbox{and} \quad k_g:= m_g(\lambda(\sigma_g)).$$  
Similarly, since $m_g(\lambda(\sigma))$ is locally nonincreasing and attains the minimum value at~$\sigma_g$, there is an open neighborhood  $U_g$ of $\sigma_g$ inside $U_a$ such that $m_g(\lambda(\sigma)) = k_g$ for all $\sigma\in U_g$.    
Let $\operatorname{GL}(n,\C)$ be the set of $n\times n$ invertible complex-valued matrices. The following fact is well-known (see~\cite{grasse2004vector} and references therein):

\begin{lemma}\label{lem:pap}
If $U_g$ is sufficiently small, then there is a real-analytic function $P: U_g \to \operatorname{GL}(n,\C)$ such that  
\begin{equation}\label{eq:blockuppertriangular}
P^{-1}AP = 
\begin{bmatrix}
	\lambda I_{k_g} & A'_{12} \\
	0 & A'_{22}	
\end{bmatrix},
\end{equation} 
where $I_{k_g}$ is the $k_g \times k_g$ identity matrix.
\end{lemma}

Note that if $k_g = n$, then $P^{-1} A P$ is simply $\lambda I_n$.

\subsubsection{Block upper triangular structures} 
Let $\Sigma'$ be a closed $d$-dimensional ball inside the open set $U_g$. For ease of notation,  we will now treat $A$, $B$, and $P$ given in Lemma~\ref{lem:pap} as  matrix-valued functions on $\Sigma'$. 
Define $A': \Sigma' \to \C^{n\times n}$ and $B': \Sigma'\to \C^{n\times m}$ as 
$$A':= P^{-1} A P \quad \mbox{and} \quad B':= P^{-1} B.$$   
Consider the linear ensemble system given by the pair $(A',B')$:
\begin{equation}\label{eq:systemsimilarity}
\dot x'(t,\sigma) = A'(\sigma) x'(t,\sigma) + B'(\sigma) u'(t), \quad \mbox{for all } \sigma\in\Sigma'.
\end{equation}
By construction, this system is obtained by first restricting the original system $(A, B)$ to $\Sigma'$ and, then, applying the similarity transformation via the  matrix-valued map~$P$. It should be clear that if $(A',B')$ is not controllable, then neither is $(A, B)$.

By Lemma~\ref{lem:pap}, $A'$ is block upper triangular. 
We will now make use of such structure to perform reduction on system $(A',B')$. Consider two cases: 

{\em Case 1: $k_g = n$.} In this case, $A' = \lambda I_n$. It follows that the dynamics of the $n$ entries $x'_i(t,\sigma)$ of system~\eqref{eq:systemsimilarity}, for $i = 1,\ldots, n$, are given by
\begin{equation*}\label{eq:scalll}
\dot x'_i(t, \sigma) = \lambda(\sigma) x'_i(t, \sigma) + b'_i(\sigma)u'(t), \quad \mbox{for all } \sigma \in \Sigma',
\end{equation*}
where $b'_i$ is the $i$th row of $B'$. Note that $(A',B')$ is controllable only if every $(\lambda, b'_i)$ is.  

{\em Case 2: $k_g < n$.} In accordance with~\eqref{eq:blockuppertriangular}, we partition $B' = [B'_1; B'_2]$ and $x'(t,\sigma) = [x'_1(t,\sigma); x'_2(t,\sigma)]$.  
Then, the dynamics of $x'_2(t,\sigma)$ are given by   
\begin{equation}\label{eq:scalllll}
\dot x'_2(t, \sigma) = A'_{22}(\sigma) x'_2(t,\sigma) + B'_2(\sigma) u'(t), \quad \forall\sigma\in \Sigma'.
\end{equation}  	
It should be clear that $(A',B')$ is controllable only if $(A'_{22},B'_2)$ is. 
System~\eqref{eq:scalllll} is not necessarily a scalar ensemble, yet the dimension of the state space of each individual system has been reduced from~$n$ to $(n - k_g)$. Iterating this reduction process in a finite number of times,  we will end up with {\em Case~1}. 

The reduction on state space is now complete. It remains to show that scalar ensemble systems over closed $d$-dimensional balls are not controllable.

  \subsection{Reduction on parameterization space}\label{ssec:reductiononparaspace}
  In this subsection, we let $\Sigma$ be a closed $d$-dimensional ball in~$\R^d$, and consider the following scalar ensemble system: 
\begin{equation}\label{eq:scalarensemble}
  \dot x(t, \sigma) = a(\sigma) x(t, \sigma) + b(\sigma)u(t), \quad \mbox{for all } \sigma \in \Sigma,  
\end{equation}
where $a:\Sigma \to \C$ and $b:\Sigma\to \C^{1\times m}$ are real-analytic functions.   
  Let $\operatorname{Re}a$ and $\operatorname{Im}a$ be the real and imaginary parts of~$a$, respectively. Define $J:\Sigma \to \R^{2\times d}$ as follows:  
  \begin{equation}\label{eq:defJacobian}
  J(\sigma):= 
  \frac{\partial}{\partial\sigma}
  \begin{bmatrix}
  	 \operatorname{Re}\, a(\sigma) \\ 
  	 \operatorname{Im}\, a(\sigma)   
  \end{bmatrix}. 
  \end{equation}
  Further, we let $$\sigma_J \in \argmax \{\operatorname{rank} J(\sigma) \mid \sigma \in \Sigma\} \quad \mbox{and} \quad k_J := \operatorname{rank} J(\sigma_J).$$   
  We establish below the following fact:

  \begin{proposition}\label{prop:jacobian}
  	If system~\eqref{eq:scalarensemble} is $\mathrm{L}^2$-controllable, then $k_J = d$.
  \end{proposition}

  Note that by construction~\eqref{eq:defJacobian}, the rank of $J(\sigma)$ is at most~$2$. Thus, a consequence of the result is that system~\eqref{eq:scalarensemble} is not controllable if $d > 2$. 
  
We establish below Proposition~\ref{prop:jacobian}. 
Because $\operatorname{rank}J(\sigma)$ is locally nondecreasing in~$\sigma$, there is an open neighborhood $U_J$ of $\sigma_J$ in $\Sigma$ such that $\operatorname{rank} J(\sigma) = k_J$ for all $\sigma\in U_J$. We can assume, without loss of generality, that $\sigma_J$ and $U_J$ are in the interior of $\Sigma$. 
The following result is a consequence of the Rank Theorem~\cite[Theorem 4.12]{lee2012introduction}:

\begin{lemma}\label{lem:ranktheorem}
	There exist a closed neighborhood $\bar U_F$ of $\sigma_J$ in $U_J$ and a  $\mathrm{C}^1$-diffeomorphism: $$\varphi: [-\epsilon_1, \epsilon_1]^{d - k_J}\times [-\epsilon_2, \epsilon_2]^{k_J}\to \bar U_F,$$ with $\epsilon_1, \epsilon_2 > 0$, such that for every $\mu_2\in [-\epsilon_2, \epsilon_2]^{k_J}$, the map $a\cdot \varphi$ is constant on the following set: $$S_{\mu_2} := [-\epsilon_1, \epsilon_1]^{d - k_J} \times \{\mu_2\}.$$ 
\end{lemma}
  	  
Each $S_{\mu_2}$ will be referred to as a slice.           
For ease of notation, we let $\bar V_F:= [-\epsilon_1, \epsilon_1]^{d - k_J}\times [-\epsilon_2, \epsilon_2]^{k_J}$ and, for clarification of presentation, we use letter $\mu$ to denote a point in the closed rectangle $\bar V_F$.             
Let $a': \bar V_F \to \C$ and $b': \bar V_F\to \C^{1\times m}$ be defined as follows:    
    $$
    a': = a \cdot \varphi \quad \mbox{and} \quad b':= b\cdot \varphi.
    $$  
    Then, the pullback of system~\eqref{eq:scalarensemble} by $\varphi$ is given by
   \begin{equation}\label{eq:foliatedchart}
   \dot x'(t, \mu) = a'(\mu)x'(t, \mu) + b'(\mu) u'(t), \quad \mbox{for all } \mu \in \bar V_F.  
   \end{equation}
  With the preliminaries above, we now prove Proposition~\ref{prop:jacobian}:
   
   {\em Proof of Proposition~\ref{prop:jacobian}.}
   By Lemma~\ref{lem:embedding}, it suffices to show that 
   if $k_J < d$, then system~\eqref{eq:foliatedchart} is not controllable. 
   Since $k_J < d$, the dimension of each slice $S_{\mu_2}$ defined in Lemma~\ref{lem:ranktheorem} is positive.     
     Let $b'_i$ be the~$i$th entry of the row vector~$b'$ and $b'_i |_{S_{\mu_2}}$ be the restriction of $b'_i$ to the slice~$S_{\mu_2}$. 
     Denote by $\mathcal{B}_{\mu_2}$ the finite-dimensional subspace of $\mathrm{L}^2(S_{\mu_2},\C)$ spanned by $b'_1 |_{S_{\mu_2}},\ldots, b'_m |_{S_{\mu_2}}$.  
   Note that $\dim\mathcal{B}_{\mu_2}$ is locally nondecreasing as a function of $\mu_2\in [-\epsilon_2, \epsilon_2]^{k_J}$. We can assume, without loss of generality, that the maximum value of $\dim\mathcal{B}_{\mu_2}$ is achieved at $\mu_2 = 0$, and let $m':= \dim\mathcal{B}_{0}$.  
    Furthermore, by decreasing the value of $\epsilon_2$, if necessary, we can assume that the first $m'$ scalar functions $b'_1 |_{\mu_2},\ldots, b'_{m'} |_{\mu_2}$ are linearly independent for all $\mu_2\in [-\epsilon_2, \epsilon_2]^{k_J}$.

   Denote by $\operatorname{P}_{\mu_2}$ the orthogonal projection of the Hilbert space ${\rm L}^2(S_{\mu_2}, \C)$ onto $\mathcal{B}^\perp_{\mu_2}$, the subspace 
    perpendicular to $\mathcal{B}_{\mu_2}$. The operator can be computed explicitly: For any $h\in {\rm L}^2(S_{\mu_2}, \C)$, we have that
   $$
    \operatorname{P}_{\mu_2}(h) = h - \sum^{m'}_{i = 1} c_i(h)  b'_i |_{S_{\mu_2}},
   $$  
   where the coefficients $c_i(h)\in \C$ are given by  
    $$
    \begin{bmatrix}
    	c_1(h) \\
    	\vdots \\
    	c_{m'}(h)
    \end{bmatrix}
	 := 
    \begin{bmatrix}
    	 \langle b'_1 |_{S_{\mu_2}},  b'_1 |_{S_{\mu_2}}  \rangle & \cdots &   \langle b'_1 |_{S_{\mu_2}},  b'_{m'} |_{S_{\mu_2}} \rangle \\
    	\vdots & \ddots & \vdots \\
    	\langle b'_{m'} |_{S_{\mu_2}},  b'_1 |_{S_{\mu_2}} \rangle & \cdots &  \langle b'_{m'} |_{S_{\mu_2}},  b'_{m'} |_{S_{\mu_2}} \rangle
    \end{bmatrix}^{-1}
    \begin{bmatrix}
    	\langle b'_1 |_{S_{\mu_2}},  h \rangle \\
    	\vdots \\
    	\langle b'_{m'} |_{S_{\mu_2}},  h \rangle
    \end{bmatrix}.
   $$
   Note that the square matrix in the above expression is Hermitian and, moreover, positive definite because $b'_1 |_{S_{\mu_2}}, \ldots, b'_{m'} |_{S_{\mu_2}}$ are linearly independent. Also, note that the coefficients $c_i(h)$ are linear in $h$ and continuous in $\mu_2$. Thus, $\operatorname{P}_{\mu_2}(h)$ is continuous in both~$h$ and~$\mu_2$.  
   
   Next, let $f: \bar V_F\to \C$ be a continuous function such that there exists at least one point $\mu_2\in [-\epsilon_2,\epsilon_2]^{k_J}$ with $f |_{S_{\mu_2}}\notin \mathcal{B}_{\mu_2}$.    
   Further, define $g: \bar V_F \to \C$ as follows:
   $$
   g(\mu_1, \mu_2):= \operatorname{P}_{\mu_2}(f |_{S_{\mu_2}}) (\mu_1,\mu_2), \quad\mbox{for all } (\mu_1, \mu_2)\in \bar V_F,
   $$
   i.e., each $g |_{S_{\mu_2}}$ is the orthogonal projection of $f |_{S_{\mu_2}}$ to $\mathcal{B}^\perp_{\mu_2}$. 
   By construction, $g$ is continuous and nonzero.

   We show below that $g$ is orthogonal to the controllable subspace $\mathcal{L}(a',b')$ associated with system~\eqref{eq:foliatedchart}. 
   By Lemma~\ref{lem:ranktheorem}, 
   the function $a'$ takes a constant value on each slice $S_{\mu_2}$; we denote the value by 
   $a'_{\mu_2}$.   
   Then, for any $i = 1,\ldots, m$ and for any $k\ge 0$,   
   $$
   \langle g, a'^k b'_i \rangle_{\bar V_F} = 
   \int_{[-\epsilon_2, \epsilon_2]^{k_J}} a'^k_{\mu_2} \langle  g |_{S_{\mu_2}},  b'_i |_{S_{\mu_2}}\rangle_{S_{\mu_2}} \mathrm{d}\mu_2 = 0,
   $$
   where the last equality holds because, by construction, $$\langle  g |_{S_{\mu_2}},  b'_i |_{S_{\mu_2}}\rangle_{S_{\mu_2}} = 0, \quad \mbox{for all } \mu_2\in [-\epsilon_2, \epsilon_2]^{k_J}.$$ 
   We thus conclude that $g$ is orthogonal to $\mathcal{L}(a',b')$, so $\mathcal{L}(a',b')$ cannot be the entire ${\rm L}^2(\bar V_F,\C)$. By Lemma~\ref{lem:controllablesubspace}, system~\eqref{eq:foliatedchart} is not controllable. 
    \endproof

 \subsection{Translation to the normal form}\label{ssec:translation}
We establish below Theorem~\ref{thm:main}.  
By the arguments in the previous subsections, we only need to consider scalar, complex linear ensemble systems over closed, two-dimensional disks~$\Sigma$: 
\begin{equation}\label{eq:reproducescalar}
\dot x(t, \sigma) = a(\sigma) x(t, \sigma) + b(\sigma)u(t), \quad \mbox{for all } \sigma \in \Sigma. 
\end{equation}  
By Proposition~\ref{prop:jacobian}, we can further assume that there is a point $\sigma_J$ in the interior of $\Sigma$ such that $\operatorname{rank} J(\sigma_J)= 2$ because otherwise, system~\eqref{eq:reproducescalar} is not controllable.  This excludes, for example, the case where $a$ is real-valued.

Again, we identify $\R^2$ with $\C$ and treat $\Sigma$ as a subset of $\C$.  
Since the Jacobian matrix $J(\sigma_J)$  has full rank, it follows from the inverse function theorem that there is an open neighborhood $U_J$ of $\sigma_J$ in the interior of $\Sigma$ such that $a: U_J\to \C$ is a real-analytic diffeomorphism between $U_J$ and its image, which we denote by $V_J$.
Let $a_J:= a(\sigma_J)$ and $D_{a_J}[R]$ be the closed disk of radius~$R > 0$ in~$\C$ centered at $a_J$. We let $R$ be sufficiently small such that $D_{a_J}[R]$ is contained in the open set $V_J$. 

Let $R$ be given as above and $D_0[R]$ be the closed disk of radius $R$ centered at~$0$.  
Now, consider the following normal form (for clarity of presentation, we use letter $\mu$ to denote a point in $D_0[R]$): 
\begin{equation}\label{eq:finalone}
 	\dot x'(t,\mu) = \mu x'(t,\mu) + b(a^{-1}(\mu + a_J)) u'(t), \quad \mbox{for all } \mu\in D_0[R]. 
\end{equation}
We establish the following result: 
 
 \begin{proposition}\label{prop:finalpiece}
 	If system~\eqref{eq:finalone} is not controllable, then neither is system~\eqref{eq:reproducescalar}.  
 \end{proposition}
 
 \begin{proof}
 For convenience, we let $b'(\mu):= b(a^{-1}(\mu + a_J))$. 
 	To establish the result, we first consider the following ensemble system as a variation of~\eqref{eq:finalone}:
 	\begin{equation}\label{eq:variationoffinalone}
 	\dot x''(t, \mu) = (\mu + a_J) x''(t, \mu) + b'(\mu) u''(t), \quad \mbox{for all } \mu \in D_0[R], 
 	\end{equation}
 	where we have replaced the ``$A$-matrix'', which is the identity function~$\omega$ in~\eqref{eq:finalone}, with the function $(\omega + a_J\mathbf{1})$ in~\eqref{eq:variationoffinalone}. Note that system~\eqref{eq:finalone} is controllable if and only if system~\eqref{eq:variationoffinalone} is. This holds because the controllable subspaces associated with the two systems are the same. Indeed, for any $k\ge 0$,  $(\omega + a_J\mathbf{1})^k$ is a linear combination of $\omega^\ell$, for $0\leq \ell \le k$. Conversely, each $\omega^k$ can be expressed as a linear combination of $(\omega + a_J\mathbf{1})^\ell$, for $0\leq \ell\leq k$. It follows that $\mathcal{L}(\omega, b') = \mathcal{L}(\omega + a_J\mathbf{1}, b')$.   
 	
 	It now suffices to show that if system~\eqref{eq:variationoffinalone} is not controllable, then neither is system~\eqref{eq:reproducescalar}. We let $\mu':= \mu + a_J$ and re-write system~\eqref{eq:variationoffinalone} as follows: 
 		\begin{equation}\label{eq:aftertranslation}
 		\dot x''(t, \mu')  =  \mu' x''(t, \mu') + b'(\mu') u''(t), \quad \mbox{for all } \mu' \in D_{a_J}[R]. 
 		\end{equation}  
 	It turns out that system~\eqref{eq:aftertranslation} is the pullback of system~\eqref{eq:reproducescalar} via the embedding $a^{-1}: D_{a_J}[R]\to \Sigma$. For this, we recall that $D_{a_J}[R]$ is contained in $V_J$, $V_J$ is the image of $U_J$ under~$a$, and $U_J$ is in the interior of $\Sigma$.   
 	Thus, by Lemma~\ref{lem:embedding}, if system~\eqref{eq:aftertranslation} is not controllable, then neither is system~\eqref{eq:reproducescalar}. 
 \end{proof}
 
A proof of Theorem~\ref{thm:main} is now at hand:
 
{\em Proof of Theorem~\ref{thm:main}.}
From Theorem~\ref{thm:specialcomplex}, normal forms are not $\mathrm{L}^2$-controllable. Thus, by Proposition~\ref{prop:finalpiece}, system~\eqref{eq:reproducescalar} is not $\mathrm{L}^2$-controllable.  The arguments in Subsections~\S\ref{ssec:reductiononstatespace} and~\S\ref{ssec:reductiononparaspace} then imply that system~\eqref{eq:modelrep} is not $\mathrm{L}^2$-controllable. Combining this with the arguments in Section~\S\ref{sec:preliminaryresult}, we complete the proof. 
\endproof

\section{Conclusions}
	 We have shown that for a linear ensemble system $(A,B)$, if its parameterization space $\Sigma$ contains an open set in $\R^d$, for $d\geq 2$, and if $A:\Sigma\to \F^{n\times n}$ and $B:\Sigma\to\F^{n\times m}$, with $\F$ being either $\R$ or $\C$,  are real-analytic at a point in $U$, then the linear ensemble system cannot be $\mathrm{L}^p$-controllable, for $2\le p \le \infty$. This negative result has implications for approximation theory and operator theory, as described in Theorem~\ref{thm:approximation} and Corollary~\ref{cor:cyclicoperator}.  
	Finally, we note that the negative result applies only to linear ensemble systems. There exist uniformly controllable ensembles of control-affine systems~\cite{chen2019structure}, with real-analytic vector fields, over multi-dimensional parameterization spaces.

\bibliographystyle{IEEEtran}
\bibliography{unctrl.bib}

\end{document}